\documentclass{article}

\usepackage[english]{babel}

\usepackage[letterpaper,top=2cm,bottom=2cm,left=3cm,right=3cm,marginparwidth=1.75cm]{geometry}

\usepackage{enumitem}
\setlist[itemize,enumerate]{noitemsep, topsep=0pt}

\date{}

\usepackage{mathtools}
\usepackage{amssymb}
\usepackage{amsthm}
\usepackage{amsmath}
\usepackage{thmtools}
\usepackage{mathtools}
\usepackage{xspace}
\usepackage[dvipsnames]{xcolor}
\usepackage[colorlinks=true, allcolors=blue]{hyperref}
\usepackage[capitalise,noabbrev]{cleveref}
\usepackage{physics}
\usepackage{dirtytalk}
\usepackage{dsfont}
\usepackage{tcolorbox}
\usepackage{thmtools}
\usepackage{thm-restate}
\tcbuselibrary{skins, breakable}

\usepackage{mdframed}
\newenvironment{mathinlay}[1][]%
  {%
  \begin{mdframed}%
  [leftmargin=.5em,rightmargin=0em,#1,bottomline=false,topline=false,rightline=false,linewidth=2pt,linecolor=black!20,innerleftmargin=.5em,innertopmargin=0pt,innerrightmargin=0pt]\begingroup\vspace*{-2pt}}%
  {\endgroup
  \end{mdframed}
  }

\newtheorem{theorem}{Theorem}[section]

\newtheorem{lemma}[theorem]{Lemma}
\newtheorem{definition}[theorem]{Definition}

\newtheorem{result}[theorem]{Result}
\newtheorem{remark}[theorem]{Remark}
\newtheorem{problem}[theorem]{Problem}

\newtheorem{question}[theorem]{Question}
\newtheorem{claim}[theorem]{Claim}
\crefname{claim}{Claim}{Claims}

\newcommand{\classname}[1]{\ensuremath{\mathsf{#1}}\xspace}
\newcommand{\QMA}{\classname{QMA}}
\newcommand{\MA}{\classname{MA}}
\newcommand{\QCMA}{\classname{QCMA}}
\newcommand{\QMAo}{\classname{QMA_1}}
\newcommand{\BQP}{\classname{BQP}}
\newcommand{\coRP}{\classname{coRP}}
\newcommand{\ZQP}{\classname{ZQP}}
\newcommand{\RQP}{\classname{RQP}}
\newcommand{\coRQP}{\classname{coRQP}}
\newcommand{\RP}{\classname{RP}}
\newcommand{\ZQPqp}{\ensuremath{\mathsf{ZQP}_{\!\!/\mathsf{qpoly}}}\xspace}
\newcommand{\BQPqp}{\ensuremath{\mathsf{BQP}_{\!\!/\mathsf{qpoly}}}\xspace}
\newcommand{\BQPp}{\ensuremath{\mathsf{BQP}_{\!\!/\mathsf{poly}}}\xspace}
\newcommand{\NP}{\classname{NP}}

\newcommand{\PreciseQCMA}{\classname{PreciseQCMA}}
\newcommand{\PreciseQMA}{\classname{PreciseQMA}}
\newcommand{\Pitwo}{\classname{\Pi_2^P}}
\newcommand{\PSPACE}{\classname{PSPACE}}

\newcommand{\tsc}[1][]{\mathcal{T}^{(#1)}}
\newcommand{\pc}{\ensuremath{\mathrm{PC}}}
\DeclarePairedDelimiter{\ceil}{\lceil}{\rceil}
\DeclarePairedDelimiter{\floor}{\lfloor}{\rfloor}
\newcommand{\poly}{\ensuremath{\mathrm{poly}}}

\newcommand{\supp}{\ensuremath{\mathrm{supp}}}
\DeclareMathOperator{\E}{\mathbb{E}}
\newcommand{\CSWAP}{\mathrm{CSWAP}}
\newcommand{\Hprop}{H_{\mathrm{prop}}}
\newcommand{\Hout}{H_{\mathrm{out}}}

\newcommand{\Hin}{H_{\mathrm{in}}}
\newcommand{\CC}{\mathbb{C}}
\newcommand{\NN}{\mathbb{N}}

\newcommand{\calC}{\mathcal{C}}
\newcommand{\calA}{\mathcal{A}}
\newcommand{\calT}{\mathcal{T}}
\newcommand{\calH}{\mathcal{H}}
\newcommand{\calG}{\mathcal{G}}

\newcommand{\FF}{\mathbb{F}}
\newcommand{\hw}{\mathrm{hw}}

\newcommand{\Dec}{\mathrm{Dec}}

\newcommand{\rc}{\mathrm{rc}}

\newcommand{\Had}{\mathrm{H}}
\newcommand{\Toff}{\mathrm{Toffoli}}
\newcommand{\Xg}{\mathrm{X}}
\newcommand{\CNOT}{\mathrm{CNOT}}

\newcommand{\CISS}{\ensuremath{\mathsf{CISS}}\xspace}
\newcommand{\Good}{\ensuremath{\mathsf{Good}}\xspace}
\newcommand{\Bias}{\ensuremath{\mathsf{Bias}}\xspace}
\newcommand{\BiasYZ}{\ensuremath{\mathsf{BiasYZ}}\xspace}
\newcommand{\adv}{\mathsf{adv}}
\newcommand{\Mult}{\mathsf{Mult}}
\newcommand{\Guesser}{\ensuremath{\mathsf{Guesser}}\xspace}

\definecolor{algbg}{HTML}{F7F7F7}
\definecolor{algframe}{HTML}{CCCCCC}
\definecolor{alglabel}{HTML}{555555}
\definecolor{inputlabel}{HTML}{888888}

\newcounter{algoctr}
\crefname{algoctr}{Algorithm}{Algorithms}
\Crefname{algoctr}{Algorithm}{Algorithms}
 
\makeatletter
\newcommand{\alglabel}[1]{%
  \setcounter{algoctr}{\value{theorem}}%
  \protected@edef\@currentlabel{\thetheorem}%
  \def\cref@currentlabel{[algoctr][\arabic{theorem}][\thesection]\thetheorem}%
  \label{#1}%
}
\makeatother
 
\newtcolorbox{algobox}[2][]{%          #1 = tcolorbox keys, #2 = title text
  enhanced,
  breakable,
  colback    = algbg,
  colframe   = algframe,
  boxrule    = 0.3pt,
  arc        = 1.5pt,
  left       = 10pt,
  right      = 10pt,
  top        = 2pt,
  bottom     = 10pt,
  fontupper  = \small,
  title      = {\refstepcounter{theorem}%
                \small\textsc{Algorithm \thetheorem}\;\;#2},
  coltitle   = alglabel,
  colbacktitle = algbg,
  toptitle   = 6pt,
  bottomtitle = 6pt,
  #1
}
 
\newcommand{\Input}[1]{%
  \par\noindent
  \textcolor{inputlabel}{\textsc{Input:}}\;\, #1
  \par\vspace{2pt}
  {\color{algframe}\hrule height 0.3pt}
  \vspace{6pt}
}

\title{En Route to a Standard \QMAo vs.\ \QCMA Oracle Separation}
\author{David Miloschewsky\thanks{Department of Computer Science, Stony Brook University, USA.  \\ Email: \texttt{\{dmiloschewsk, supartha\}@cs.stonybrook.edu}.} \and Supartha Podder\footnotemark[1] \and  Dorian Rudolph\thanks{Department of Computer Science and Institute for Photonic Quantum Systems (PhoQS),
Paderborn University, Germany. Email: \texttt{dorian.rudolph@upb.de}.}}

\begin{document}
\maketitle

\begin{abstract}
    We study the power of quantum witnesses under perfect completeness. We construct a classical oracle relative to which a language lies in \QMAo but not in \QCMA when the \QCMA verifier is only allowed polynomially many adaptive rounds and exponentially many parallel queries per round. Additionally, we derandomize the permutation-oracle separation of Fefferman and Kimmel, obtaining an in-place oracle separation between \QMAo and \QCMA. Furthermore, we focus on \QCMA and \QMA with an exponentially small gap, where we show a separation assuming the gap is fixed, but not when it may be arbitrarily small. Finally, we derive consequences for approximate ground-state preparation from sparse Hamiltonian oracle access, including a bounded adaptivity frustration-free variant.
\end{abstract}

\section{Introduction}

The complexity class \QMA (Quantum Merlin-Arthur, a quantum analogue of $\NP$) was formally introduced by Kitaev~\cite{kitaev1999quantum} as the class of languages for which a polynomial-time quantum verifier can be convinced to accept yes-instances with high probability using a polynomial-size quantum witness, while rejecting no-instances regardless of the witness supplied. A natural variant of \QMA, introduced by Aharonov and Naveh~\cite{aharonov2002quantum}, is \QCMA, where the verifier is quantum but the witness is restricted to be classical.

It is immediate that $\QCMA \subseteq \QMA$. A central question in quantum complexity theory is whether this inclusion is strict, i.e., whether quantum witnesses provide additional computational power over classical ones in the non-interactive setting. Despite significant effort, the status of this question in the unrelativized world remains open. Early intuition by Aharonov and Naveh~\cite{aharonov2002quantum} suggested that quantum witnesses might not be inherently more powerful, for instance, if ground states of local Hamiltonians were to admit efficient classical descriptions, then $\QMA = \QCMA$.

Considerable progress has been made in understanding this question in standard and relativized worlds. 
Watrous~\cite{watrous2000succinct} showed that the Group Non-Membership problem separates \QMA from $\MA$ and was one of the first candidates for understanding the power of classical and quantum witnesses~\cite{aaronson2007quantum, LGNT25}. Aaronson and Kuperberg~\cite{aaronson2007quantum} gave the first separation between \QMA and \QCMA, albeit relative to a quantum oracle. Fefferman and Kimmel~\cite{FK18} later obtained separations using randomized in-place permutation oracles, and several subsequent works proposed alternative oracle constructions and candidate problems aimed at clarifying the source of quantum advantage~\cite{lutomirski2011component,laracuente2021quantum,myrisiotis2016quantum}. More recent efforts pursued separations in increasingly standard oracle models~\cite{li2024classical, natarajan2024distribution, BK24, zhandry2025toward, liu2025qma}. In parallel, a line of work has investigated this question through the lens of uncloneability~\cite{nehoran2024computational, cakan2025public, broadbent2025role, CKP25}, structural properties of \QMA and its variants~\cite{grilo2015qma,dziemba2017robustness,nagaj2021pinned}, as well as in communication complexity~\cite{raz2004power,klauck2014two,klauck2017complexity}. The longstanding question of whether \QMA and \QCMA can be separated relative to a standard classical oracle was ultimately resolved by~\cite{BHNZ26,BHV26}. 

Thus a classical counterpart of the result of Aaronson and Kuperberg~\cite{aaronson2007quantum} was finally found. However, as noted by~\cite{BHNZ26} (and also not achieved by~\cite{BHV26}), neither of these classical oracle separations achieve perfect completeness. In contrast, \cite{aaronson2007quantum} actually separates the class \QMAo from \QCMA with respect to a quantum oracle.

\paragraph*{\QMA with perfect completeness.}\QMAo is the variant of \QMA in which the verifier must accept every YES instance with probability \emph{exactly} $1$ (while still accepting NO instances with probability at most $1/2$). The class was formalized in~\cite{KSV02} and gained prominence through the \QMAo-completeness of Quantum $k$-SAT~\cite{Bravyi2006,GossetNagaj2013}. Whether $\QMA = \QMAo$ remains open; Aaronson~\cite{Aaronson2008} gave a quantum oracle separation, and the answer is known to depend delicately on the verifier's gate set. Amplification and completeness techniques for one-sided-error \QMA have been developed in~\cite{NWZ2009,KLN2015}. In contrast, the classical-witness analogue satisfies $\QCMA = \QCMA_1$ under standard gate sets~\cite{JKNN2011}. This leaves open the following natural question.

\begin{question}\label{question1}
     Is there a \emph{standard classical} oracle separation between  \QMAo and \QCMA?
\end{question}

\medskip

\noindent As our first result, we give a bounded adaptivity analog of \Cref{question1}.

\begin{result}[\Cref{thm:bounded_adaptivity}]\label{res:bounded_adaptivity}
    For every polynomial $R$, there exists a classical oracle relative to which
    $\QMAo \not\subseteq \QCMA$, assuming the \QCMA verifier is restricted to $R$ rounds of adaptive queries. Additionally, the language is in \coRP with $R+1$ adaptive one-query rounds.
\end{result}

By bounded adaptivity, we mean that the oracle calls are grouped into rounds. Each round of queries is performed in parallel and may only depend on the results of previous queries. Therefore, the restriction is on the depth of the oracle dependence, not on the total number of queries. In \cref{res:bounded_adaptivity}, the \QCMA lower bound applies even when each round makes exponentially many queries. The regime described above is the same as in~\cite{BK24}, who separated \QCMA from \QMA (but not \QMAo) with $o(\log n /\log \log n)$ adaptive rounds. In contrast, \cref{thm:bounded_adaptivity} separates \QMAo from \QCMA bounded to $R$ rounds, where $R$ is any fixed polynomial. On the other hand, the bounded round restriction is necessary in our problem as a \coRP algorithm solves the problem with $R+1$ rounds, while the problem in~\cite{BK24} was used by~\cite{BHV26} for the classical oracle separation.

\paragraph*{In-place oracle separation.} As discussed before, in the permutation oracle model, Fefferman and Kimmel~\cite{FK18} separated \QMA from \QCMA using \emph{randomized} in-place permutation oracles. Despite the extensive line of work on the \QMA versus \QCMA question that followed~\cite{li2024classical, natarajan2024distribution, BK24, zhandry2025toward, liu2025qma, BHNZ26, BHV26}, derandomizing the~\cite{FK18} construction has remained elusive, and none of these works achieve perfect completeness in the in-place model. Furthermore, the in-place oracle model has been studied in other contexts as well~\cite{Aar02, KKVB02, BFM23, HRY25}. This motivates the following natural question.

\begin{question}\label{question2}
Can randomization be removed from the in-place oracle separation of Fefferman and Kimmel~\cite{FK18}, and can the separation be strengthened to \QMAo versus \QCMA? 
\end{question} 

\medskip

\noindent We answer both parts affirmatively.

\begin{result}[\Cref{thm:inplace-qma-qcma}]\label{res:inplace-qma-qcma}
There exists a deterministic in-place permutation oracle relative to which $\QMAo \not\subseteq \QCMA$, even when the \QCMA verifier is allowed $2^{o(n)}$ rounds of adaptivity and exponentially many queries per round.
\end{result}

Viewed as a query problem, \cref{res:inplace-qma-qcma} is a set-preimage sampling lower bound. Formally defined in \cref{prob:perm_pointer_chasing}, the problem is to estimate the parity bias of the preimage set $S_b = \pi^{-1}([N])$ where $\pi$ is the given permutation. The \QMAo witness is simply $\ket{S_b}$. Without a witness, one may sample $S_b$ by drawing some $j\in [N]$ and applying $\pi$ for $b-1$ steps. This differs from the standard permutation inversion problem, where given some $y$, one must find $\pi^{-1}(y)$, which has query complexity $\Theta(\sqrt{N})$ both in the standard \textup{XOR} and in-place oracle models~\cite{Gro96,BBBV97,FK18,HRY25}. Thus \cref{res:inplace-qma-qcma} shows that even when given a \emph{classical witness}, the simple $b-1$ round sampling algorithm is optimal.

\paragraph*{Witness verification with an exponentially small gap.}A natural question in the study of verification classes with error is how they behave under various completeness-soundness gaps. For \QMA and \QCMA, standard error-reduction~\cite{MW05} shows that any inverse-polynomial gap can be amplified to a constant, meaning the specific choice of parameters is irrelevant. However, when the gap is allowed to shrink to $2^{-\poly(n)}$, no error-reduction techniques are known. The resulting classes \PreciseQMA and \PreciseQCMA have been characterized as \PSPACE~\cite{FL18} and $\NP^{\mathsf{PP}}$~\cite{MN17, GSSSY22}, respectively.

\begin{question}\label{que:precise_gap}
    Is there a classical oracle separating \QMA from \QCMA when the completeness-soundness gap is exponentially small?
\end{question}

\noindent We partially answer this question. Specifically, by tweaking the original proof technique of~\cite{BHV26}, we show that when the gap is fixed to $O(2^{-n})$, the classes may be separated. However, \PreciseQCMA, which considers the union over all $2^{-\poly(n)}$ gaps, solves the problem.

\begin{result}[\Cref{thm:precise_separation}]\label{res:precise_sep}
There exists a classical oracle relative to which a problem lies in $\QMA(2^{-n})$ and $\PreciseQCMA$, but not in $\QCMA(2^{-n})$.
\end{result}

The \PreciseQCMA inclusion holds due to the fact that the problem is in \Pitwo, which is in $\classname{NP}^\classname{PP}$ due to Toda's theorem~\cite{Tod91}. In particular, this provides evidence against an efficient gap amplification procedure for \PreciseQMA and \PreciseQCMA.

\medskip

\paragraph*{Circuit complexity of sparse Hamiltonian ground states.}The No Low-Energy Trivial States (NLTS) theorem of Anshu, Breuckmann, and Nirkhe~\cite{ABN23}, resolving a conjecture of Freedman and Hastings~\cite{FH14}, establishes that there exist families of local Hamiltonians whose low-energy states cannot be prepared by constant-depth quantum circuits. In an oracle setting, one can ask whether sparse Hamiltonians given via Hamiltonian query access $O_H^{\rc}$ have approximate ground states that provably require large circuits to prepare.

\begin{question}\label{question4}
    Are there families of sparse Hamiltonians, given via oracle access, whose approximate ground states provably require large circuits to prepare?
\end{question}

A natural question in a similar spirit is whether the oracle separations between \QMA and \QCMA yield circuit complexity lower bounds for ground states of sparse Hamiltonians given via Hamiltonian query access $O_H^{\rc}$. To the best of our knowledge, this connection has not been made explicitly in prior work, although it follows relatively directly from applying the circuit-to-Hamiltonian construction~\cite{KSV02} to the \QMA verifiers in~\cite{BHNZ26,BHV26}.

\begin{result}[\Cref{thm:sparse_ham_hard}]\label{res:ham}
    There exists a family of $O(1)$-sparse Hamiltonians $\{H_i\}$ on $n_i$ qubits such that any state preparable by a circuit of size $2^{o(n_i^{1/3}/\log n_i)}$ which has access to $H$ via the oracle $O_{H_i}^{\rc}$ has energy at least $\lambda_0(H_i) + 1/\poly(n_i)$.
\end{result}

Furthermore, using \cref{res:bounded_adaptivity}, we are able to extend this to frustration-free Hamiltonians, except the circuit is restricted in adaptivity (see \cref{thm:bounded_adaptivity_frustration}). Note that extending this without the restriction would resolve \cref{question1}.

\paragraph*{Techniques.}

The problem used to obtain \cref{res:inplace-qma-qcma} is a pointer-chasing permutation oracle instance. The permutation consists of disjoint chains through buckets $I_1,\dots I_b$ with the first bucket being $[N]$ and the last bucket mapping to $[N]$. The problem is to decide whether parity of the items in $I_b$ which map to $[N]$ is $0$ or $1$, given the promise that it is either $0$ in all items, or only over half of them (for a formal definition, see \cref{prob:perm_pointer_chasing}). In a YES instance, the \QMAo witness is the uniform state on $S_b=\pi^{-1}([N])$, which is used to test the parity of $S_b$ and whether it maps back to $[N]$, both of which are checked with perfect completeness.

The \QCMA lower bound is as follows. Given a fixed verifier and classical witness, we record a transcript $\tsc[r]$ of all the points of the oracle which the algorithm \say{knows} after $r$ adaptive rounds of queries. Each round, $\tsc[r]$ may only add points by either querying them with heavy weight or from the previous step in the chain. We describe a canonical permutation $\pi^\prime$ based on $\tsc[r]$ and show that the verifier's state using $\pi$ is close to that with $\pi^\prime$. Assuming $\tsc[r]$ satisfies several conditions which we call \say{good}, which occurs with high probability, each YES instance may be transformed into a NO instance with the same canonical permutation by altering the last bucket $I_b$. A union bound over witnesses and diagonalization finishes the lower bound.

\cref{res:bounded_adaptivity} uses the tools developed above. First, the \QCMA lower bound applies using the standard \textup{XOR} oracle. Second, the \QMAo algorithm is obtained by considering the algorithm which uses $R+1$ rounds of queries and parallelizing it using the Feynman-Kitaev history state. As we show, any $R$-round \QMA (and \QMAo) algorithm may be parallelized to a single-round verifier with $O(R^4)$ parallel oracle queries.

We note that the oracle separations underlying \cref{res:bounded_adaptivity,res:inplace-qma-qcma} were developed independently of the two recent classical-oracle separations between \QMA and
\QCMA~\cite{BHNZ26,BHV26}. On the other hand, the results in \cref{sec:extensions} directly use the machinery developed in~\cite{BHV26}. In \cref{res:precise_sep}, we solely modify the bias parameter. Furthermore, \cref{res:ham} is obtained by applying the circuit-to-Hamiltonian construction~\cite{KSV02} to the \QMA verifier in~\cite{BHV26} and using sparse-Hamiltonian energy estimation. In the frustration-free case, we use the \QMAo verifier for \cref{prob:perm_pointer_chasing}.

\paragraph*{Acknowledgements.}We thank Srinivasan Arunachalam, Sevag Gharibian, Srijita Kundu, Sinan Oral, and Nengkun Yu for helpful discussions.
DR acknowledges funding from the DFG Priority Programme 2514 (grant number 563388236).
This work was done in part while the authors were visiting the Simons Institute for the Theory of Computing.

\section{Preliminaries}

We assume the reader is familiar with standard quantum computing and complexity theory background. For introductions, see~\cite{NC10, AB09}.
We use the following notation. We use $[a]$ to denote the set $\{1,\dots a\}$ and $[a,b]$ to denote $\{a, a+1,\dots ,b\}$. Given a set of integers $S$, let $S^{\textup{even}}$ and $S^{\textup{odd}}$ denote the subset of $S$ which only contains its even or odd elements, respectively. For any finite non-empty set $S$, we let 
\begin{align*}
    \ket{S} = \tfrac{1}{\sqrt{\abs{S}}}\sum_{s\in S} \ket{s}.
\end{align*}
Given a matrix $H$, we let $\lambda_0(H)$ denote the smallest eigenvalue of $H$.
Let us describe the oracle models we will use. For a string $x\in\{0,1\}^n$, the standard \textup{XOR} oracle is,
\begin{align*}
    O_x\ket{i,b} \coloneq \ket{i, b\oplus x_i},
\end{align*}
where $i\in[n]$ and $b\in\{0,1\}$. We call this a \emph{classical oracle} as, when $x$ represents a function $f:\{0,1\}^{\log n}\to \{0,1\}$, $O_x$ is the associated unitary applied.
A \emph{forward in-place oracle} for a permutation $\pi$ is the oracle $O_\pi$ defined as,
\begin{align*}
    O_\pi\ket{x} = \ket{\pi(x)}.
\end{align*}
We emphasize that we do not provide access to the inverse of $O_\pi$. Lastly, we will use the Hamiltonian query model. We say a Hamiltonian $H$ is $d$-\emph{sparse} if for each row, there are at most $d$ non-zero terms, while $H = \sum_i H_i$ is \emph{frustration-free} if $H_i \succeq 0$ and $\lambda_0(H) = 0$.
Consider a $d$-sparse Hamiltonian $H\in \CC^{M\times M}$. Let $i,j \in [M]$, $k\in[d]$ and $f(i,k)$ be the $k$th column index of a non-zero term in row $i$ of $H$. The $O_H$ and $O_f$ oracles are defined as,
\begin{align*}
    O_H\ket{i,j,z} &\coloneq \ket{i,j,z\oplus H_{i,j}}\\
    O_f \ket{i,k} &\coloneq \ket{i,f(i,k)}.
\end{align*}
We call the pair $O_H$ and $O_f$ the row-column oracle gates $O_H^{\rc}$.

An algorithm $A$ with $R$ adaptive rounds of $q$ queries to an arbitrary oracle $O$ is a sequence of unitaries $A = (U_0, U_1,\dots, U_R)$ interlaced with parallel calls to $O$. On some input $\ket{\psi}$ it behaves as,
\begin{align*}
    A\ket{\psi} \coloneq U_R O^{\otimes q}U_{R-1} O^{\otimes q}\dots U_1 O^{\otimes q} U_0.
\end{align*}

\subsection{Complexity Classes}

\begin{definition}[Variants of \BQP]
    A language $L\subseteq \{0,1\}^*$ is in $\BQP(c,s)$ where $c\geq s$ if there exists a uniform polynomial-time family of quantum circuits $Q=\{Q_n\}$ such that,
    \begin{align*}
        x\in L &\implies \Pr[Q(x) = 1] \geq c\\
        x\not\in L &\implies \Pr[Q(x) = 1] \leq s.
    \end{align*}
    We define $\BQP = \BQP(\tfrac{2}{3},\tfrac{1}{3})$, $\RQP = \BQP(\tfrac{1}{2}, 0)$ and $\ZQP = \RQP \cap \coRQP$.
\end{definition}

Equivalently, \ZQP is the class of zero-error quantum algorithms which may output $1,0$ or \say{maybe}, where the probability of \say{maybe} is at most $\tfrac{1}{2}$ and the algorithm never answers incorrectly. 
Lastly, \RP is defined analogously to \RQP, with the quantum circuit replaced by a probabilistic Turing machine.

\begin{definition}[Witness classes]
    A language $L\subseteq\{0,1\}^*$ is in $\QMA(c,s)$ where $c\geq s$ if there exist a uniform polynomial-time quantum verifier $V$ and a polynomial $p$ such that for every $x\in\{0,1\}^*$,
    \begin{align*}
        x\in L &\implies \exists\ket{\psi}\in(\CC^2)^{\otimes p(\abs{x})}\colon \Pr\left[V(x,\ket{\psi})=1\right]\geq c\\
        x\notin L &\implies \forall\ket{\psi}\in(\CC^2)^{\otimes p(\abs{x})}\colon \Pr\left[V(x,\ket{\psi})=1\right]\leq s.
    \end{align*}
    $\QCMA(c,s)$ is the restriction of $\QMA(c,s)$ in which the witness is a classical string $y\in\{0,1\}^{p(\abs{x})}$.
    We define the following classes,
    
    \begin{align*}
        \QMA &= \bigcup_{c-s \geq \tfrac{1}{\poly(n)}} \QMA(c,s),\\
        \QCMA &= \bigcup_{c-s \geq \tfrac{1}{\poly(n)}} \QCMA(c,s),\\
        \QMAo &= \bigcup_{s\leq 1-\frac{1}{\poly(n)}} \QMA(1,s),\\
        \PreciseQMA &= \bigcup_{c-s \geq 2^{-\poly(n)}} \QMA(c,s),\\
        \PreciseQCMA &= \bigcup_{c-s \geq 2^{-\poly(n)}} \QCMA(c,s),
    \end{align*}
    where the parameters $c,s$ are functions of the input length.
\end{definition}

By standard completion and soundness gap amplification techniques~\cite{MW05}, $\QMA = \QMA(\tfrac{2}{3}, \tfrac{1}{3})$, $\QCMA = \QCMA(\tfrac{2}{3}, \tfrac{1}{3})$ and $\QMAo = \QMA(1, \tfrac{1}{2})$. However, the same techniques are not known to apply to \PreciseQMA and \PreciseQCMA.

\begin{definition}[Bounded adaptivity oracle classes]
    Let $O=\{O_n\}_n$ be a family of oracles and $R,q\colon\NN\to\NN$ be query bounds. A language $L\subseteq\{0,1\}^*$ is in $\QMA^O[R,q]$ if there exist a uniform quantum verifier $V$ and a polynomial $p$ such that $V$ makes at most $R(\abs{x})$ adaptive rounds of $q(\abs{x})$ parallel queries to $O$ and for every $x\in\{0,1\}^*$,
    \begin{align*}
        x\in L &\implies \exists\ket{\psi}\in(\CC^2)^{\otimes p(\abs{x})}\colon \Pr\left[V^{O}(x,\ket{\psi})=1\right]\geq\frac{2}{3},\\
        x\notin L &\implies \forall\ket{\psi}\in(\CC^2)^{\otimes p(\abs{x})}\colon \Pr\left[V^{O}(x,\ket{\psi})=1\right]\leq\frac{1}{3}.
    \end{align*}
    The classes $\QCMA^O[R,q]$ and $\QMAo^O[R,q]$ are defined analogously.
\end{definition}

\begin{definition}[Advice classes]
    A language $L\subseteq \{0,1\}^*$ is in $\BQPqp$ if there is a polynomial-time family of quantum circuits $Q=\{Q_n\}$, a polynomial $p(n)$ and a family of quantum states $\{\ket{\psi_n} \in (\CC^2)^{\otimes p(n)}\}$ such that,
    \begin{align*}
        x\in L &\implies \Pr[Q(x, \ket{\psi_{\abs{x}}}) = 1] \geq \frac{2}{3},\\
        x\notin L &\implies \Pr[Q(x, \ket{\psi_{\abs{x}}}) = 1] \leq \frac{1}{3}.
    \end{align*}
    \BQPp is defined analogously, except the advice states $\ket{\psi_n}$ are replaced by classical strings $y\in \{0,1\}^{p(n)}$. Similarly, \ZQPqp is defined as \BQPqp with the zero-error semantics of \ZQP.
\end{definition}

\subsection{Tools}

We will use a parallel-query hybrid method~\cite{BBBV97}.

\begin{restatable}[Parallel-query hybrid method]{lemma}{lemParallelHybrid}\label{lem:parallel_hybrid}
    Let $\pi, \pi^\prime$ be permutations over $[M]$, $\ket{\phi}$ an arbitrary state over $q$ query registers, $\Pi$ be projector $\Pi\coloneq \sum_{x: \pi(x)\neq \pi^\prime(x)} \ketbra{x}{x}$ and $\Pi_j$ be $\Pi$ over the query register $j\in[q]$. Furthermore, define $p_j = \bra{\phi} \Pi_j \ket{\phi}$. Then,
    \begin{align*}
        \norm{\left(O_\pi^{\otimes q} - O_{\pi^\prime}^{\otimes q}\right)\ket{\phi}} \leq 2\sqrt{\sum_{j\in[q]} p_j}.
    \end{align*}
\end{restatable}
\begin{proof}
    Let $\Pi^\prime \coloneq I - \prod_{j=1}^q (I - \Pi_j)$. A basis state lies in $\ker(\Pi^\prime)$ if and only if all query inputs agree between $\pi$ and $\pi^\prime$. Therefore,
    \begin{align*}
        \left(O_\pi^{\otimes q} - O_{\pi^\prime}^{\otimes q}\right)\ket{\phi} = \left(O_\pi^{\otimes q} - O_{\pi^\prime}^{\otimes q}\right)(\Pi^\prime + (I - \Pi^\prime))\ket{\phi} = \left(O_\pi^{\otimes q} - O_{\pi^\prime}^{\otimes q}\right) \Pi^\prime\ket{\phi},
    \end{align*}
    as $(I - \Pi^\prime)\ket{\phi}$ lies in $\ker(\Pi^\prime)$. Therefore,
    \begin{align*}
        \norm{\left(O_\pi^{\otimes q} - O_{\pi^\prime}^{\otimes q}\right)\ket{\phi}} &\leq \norm{O_\pi^{\otimes q} - O_{\pi^\prime}^{\otimes q}}\norm{\Pi^\prime \ket{\phi}}\\
        &\leq 2 \sqrt{\bra{\phi} \Pi^\prime \ket{\phi}}\\
        &\leq 2 \sqrt{\sum_{j\in [q]} p_j},
    \end{align*}
    where the last inequality is an application of the union bound.
\end{proof}

\begin{restatable}{lemma}{lemTooHeavy}\label{lem:too_heavy}
    Let $U$ be a finite set such that $\abs{U}=B$ and let $S\subset U$ be uniformly random subset of size $N$. Letting $\mu$ be an arbitrary distribution on $U$ and $H = \{x\in U: \mu(x) \geq \epsilon/N\}$, 
    \begin{align}
        \Pr_{S \sim \textup{unif}}[\abs{S \cap H} \geq \delta N] \leq \left(\frac{eN}{\delta \epsilon B} \right)^{\delta N}.\label{eq:too-heavy}
    \end{align}
\end{restatable}
\begin{proof}
    Set $t = \lceil \delta N\rceil$.
  We have $h\coloneq|H| \le N/\epsilon$ as $1 \ge \mu(H) \ge |H|\cdot \epsilon/N$.
  To prove~\cref{eq:too-heavy}, we upper bound the number of $S$ satisfying the condition on $S\cap H$, which we call bad. These can be constructed by partitioning $S=S_1\sqcup S_2$ with $S_1 \in \binom{H}{t}$ and $S_2 \in \binom{U\setminus S_1}{N-t}$.
  There are at most $\binom{h}{t}$ choices for $S_1$, $\binom{B-t}{N-t}$ choices for $S_2$ and $\binom{B}{N}$ choices for $S$. Note that this is overcounting the bad instances of $S$ such that $\abs{S \cap H} > t$.
  The probability of choosing a bad $S$ is bounded by
  \begin{align*}
    \Pr_{S}\bigl[|S\cap H| \ge t\bigr] &\le \binom{h}{t}\frac{\binom{B-t}{N-t}}{\binom{B}{N}} \le \left(\frac{eh}{t}\right)^t \left(\frac{N}{B}\right)^t \le \left(\frac{eN}{\epsilon\delta N}\right)^t\left(\frac{N}{B}\right)^t \le \left(\frac{eN}{\delta\epsilon B}\right)^{\delta N},
  \end{align*}
  using the standard bound $\binom{h}{t} \le (\tfrac{eh}{t})^t$ and 
  \begin{align*}
    \frac{\binom{B-t}{N-t}}{\binom{B}{N}} &= \frac{(B-t)!}{(N-t)!(B-N)!}\cdot \frac{N!(B-N)!}{B!} = \frac{N!}{(N-t)!}\cdot \frac{(B-t)!}{B!}\\
    &= \frac{(N)_t}{(B_t)}
    = \prod_{i=0}^{t-1} \frac{N-i}{B-i} \le \left(\frac{N}{B}\right)^t,
  \end{align*}
  where $(N-i)/(B-i) \le N/B$ because $(N-i)B\le N(B-i)\Longleftrightarrow B\ge N$.
\end{proof}

\begin{restatable}[Controlling oracles]{lemma}{lemControlOracle}\label{lem:control-oracle}
  Let $U$ be an $n$-qubit unitary.
  We can implement $U' = \ketbra0\otimes I + \ketbra1\otimes U$ with one application of $U$, given access to an ancilla $\ket{\mu}$ with $U\ket{\mu}=\ket{\mu}$.
  The ancilla is not disturbed.
\end{restatable}
\begin{proof}
  Denote the registers of $U'$ by $\calC$ (control) and $\calT$ (target), and let $\ket{\mu}$ be given in ancilla register $\calA$.
  Let $V = \CSWAP_{\calC,\calT,\calA} \cdot U_\calA\cdot \CSWAP_{\calC,\calT,\calA}$.
  Then $V$ acts as follows on input state $\ket{\psi}\in \CC^{2^n}$,
  \begin{align*}
    \ket{0}_\calC\ket{\psi}_{\calT}\ket{\mu}_\calA  &\xmapsto{\CSWAP}  \ket{0}\ket{\psi}\ket{\mu} \xmapsto{U_{\calA}}  \ket{0}\ket{\psi}\ket{\mu} \\ &\xmapsto{\CSWAP}  \ket{0}\ket{\psi}\ket{\mu}= (U'\ket{0}\ket{\psi})\ket{\mu}\\
    \ket{1}_\calC\ket{\psi}_{\calT}\ket{\mu}_\calA &\xmapsto{{\CSWAP}} \ket{1}\ket{\mu}\ket{\psi} \xmapsto{{U_{\calA}}} \ket{1}\ket{\mu}(U\ket{\psi})\\ &\xmapsto{{\CSWAP}} \ket{1}(U\ket{\psi})\ket{\mu}\; = (U'\ket{1}\ket{\psi})\ket{\mu}.
  \end{align*}

  Thus, for $\ket{\phi}\in\CC^{2^{n+1}}$, we have $V(\ket{\phi}\ket{\mu}) = (U'\ket{\phi})\ket{\mu}$.
\end{proof}

\begin{remark}\label{rem:permutation}
  If $U$ is an $n$-qubit permutation, then we can use $\ket{\mu}=\ket{+}^{\otimes n}$ in \cref{lem:control-oracle} as $U\ket{\mu}=\ket{\mu}$.
\end{remark}

Notice that both the in-place and \textup{XOR} oracles are permutations.

Combining the quantum phase estimation algorithm~\cite{Kit95} with algorithms simulating $H$ via $O_H^{\rc}$~\cite{BCK15}, we get the following result.

\begin{restatable}{theorem}{thmQpeEnergy}\label{thm:qpe_energy}
    Let $H \succeq 0$ be an $O(1)$-sparse Hamiltonian on $n$ qubits whose entries are specified up to $O(\log n)$ bits, and let $\ket{\psi}$ be a state. Then, for every $\epsilon,\delta > 0$, there exists a quantum algorithm which uses $K= O(\tfrac{d^2}{\epsilon^2} \norm{H}_{\textup{max}}^2 \log \tfrac{1}{\delta})$ copies of $\ket{\psi}$, makes $\tilde{O}(\tfrac{d^3}{\epsilon^3} \norm{H}_{\textup{max}}^3 \log \tfrac{1}{\delta} )$ queries to $O_H^\rc$ and outputs an estimate $E$ such that,
    \begin{align*}
        \Pr[\abs{E - \bra{\psi} H \ket{\psi}} \leq \epsilon] \geq 1-\delta.
    \end{align*}
\end{restatable}
\begin{proof}
    Let $\Lambda = \norm{H}$ and write $\ket{\psi} = \sum_j \alpha_j \ket{\lambda_j}$ in terms of the eigenbasis of $H$. We run QPE on the unitary $U\coloneq e^{2\pi i H / (2\Lambda)}$ with accuracy $\eta = \tfrac{\epsilon}{4\Lambda}$ and failure probability $p = \tfrac{\epsilon}{4\Lambda}$. Let $\tilde{\phi}$ be the phase estimate and $\tilde{\lambda} \coloneq 2\Lambda \tilde{\phi}$. Conditioned on some $\lambda_j$, we have that,
    \begin{align*}
        \Pr [\abs{\tilde{\lambda} - \lambda_j} \leq \frac{\epsilon}{2}] \geq 1 - p.
    \end{align*}
    Therefore, by taking the expectation,
    \begin{align*}
        \E[\abs{\tilde{\lambda} - \lambda_j}| j] \leq \frac{\epsilon}{2} (1-p) + \Lambda \cdot p \leq \frac{3\epsilon}{4},
    \end{align*}
    where we split the equation based on whether $\abs{\tilde{\lambda} - \lambda_j} \leq \tfrac{\epsilon}{2}$ or not. By averaging over all $j$,
    \begin{align*}
        \abs{\E[\tilde{\lambda}] - \bra{\psi} H \ket{\psi}} &= \abs{\sum_j \abs{\alpha_j}^2 \E[\tilde{\lambda}|j]} - \sum_j \abs{\alpha_j}^2 \lambda_j\\
        &\leq \abs{\sum_j \abs{\alpha_j}^2 \abs{\E[\tilde{\lambda}|j] - \lambda_j}}\\
        &\leq \sum_j \abs{\alpha_j}^2 \cdot \frac{3\epsilon}{4} \leq \frac{3\epsilon}{4}.
    \end{align*}
    We will run the QPE procedure $K= O(\Lambda^2\log(1/\delta)/\epsilon^2)$ times and let the output $E$ be the average over all $\tilde{\lambda}$. By Hoeffding's inequality, we have $\Pr[\abs{E - \E[\tilde{\lambda}]} > \tfrac{\epsilon}{4}] \leq \delta$. Therefore, with probability $1-\delta$,
    \begin{align*}
        \abs{E - \bra{\psi} H \ket{\psi}} \leq \abs{E - \E[\tilde{\lambda}]} + \abs{\E[\tilde{\lambda}] - \bra{\psi} H \ket{\psi}} \leq \epsilon
    \end{align*}

    Let us calculate the query complexity. QPE with accuracy $\eta$ uses $O(\log(\Lambda/\epsilon))$ control qubits and for $t\in[0,m-1]$, the unitaries $U^{2^t}$. This requires simulating $H$ for time $\tau_t = \tfrac{\pi 2^t}{\Lambda}$. By~\cite{BCK15}, this requires $\tilde{O}(\norm{H}_{\textup{max}}\cdot d \cdot \tau_t)$ queries to $O_H^\rc$. As $\sum_{t\in[0,m-1]}2^t = 2^m - 1 = O(\tfrac{\Lambda}{\epsilon})$, the procedure for each copy $\ket{\psi}$ requires $\tilde{O}(\tfrac{d}{\epsilon} \norm{H}_{\textup{max}})$ queries. As $\Lambda \leq d \cdot \norm{H}_{\textup{max}}$, for $K$ parallel copies, we have that the total number of queries is $\tilde{O}(\tfrac{d^3}{\epsilon^3} \norm{H}_{\textup{max}}^3 \log \tfrac{1}{\delta} )$.
\end{proof}

\begin{lemma}[{Implicit in \cite[Lemma 6.1]{Rud25}}]\label{lem:kernel-tester}
    Let $H \succeq 0$ be an $O(1)$-sparse Hamiltonian on $n$ qubits with integer entries bounded by $O(1)$.
    There is an algorithm using $\poly(n)$ $\calG$ gates and $O(1)$ queries to $O_H^\rc$. On input $\ket{\psi}\in \CC^{2^n}$, the algorithm rejects with probability $q$, where $q\in[\alpha E, \beta E]$ for constants $\alpha,\beta > 0$ and $E = \ev{H}{\psi}$.
\end{lemma}
Notably, the algorithm of \cref{lem:kernel-tester} has perfect completeness and the rejection probability is proportional to the energy of the input state.
Thus, the rejection probability serves as a constant-factor multiplicative estimate of $E$.

\section{Results}

Let us formally list the results of the paper.

\subsection{In-place oracle separation via pointer-chasing}

Let us begin by fixing multiple variables we will use throughout the section. Fix some constants $l\geq 4$ and $\alpha\in (0,\tfrac{1}{2})$. For any $n\in \NN$, let $N\coloneq 2^n$, $M\coloneq 2^{ln}$, $b=\ceil{2^{\alpha n}}$ and $B\coloneq 2\floor{\tfrac{M - N}{2(b-1)}}$.\footnote{Although all the variables are functions in terms of $n$, we omit writing the input $n$.} Lastly, for $i\in [b]$, define the buckets $I_i$ as,
\begin{align*}
    I_i = \begin{cases}
        [N] &\text{if $i = 1$}\\
        [N + (i-2)B + 1, N+ (i-1)B] &\text{if $i\geq 2$.}
    \end{cases}
\end{align*}

Therefore, for $i\geq 2, \abs{I_i}= B$. Lastly, $I_{\textup{junk}} = [M]\setminus \bigcup_i I_i$. Let us define the problem we will analyze.

\begin{problem}[Permutation pointer-chasing]\label{prob:perm_pointer_chasing}
    For $i\in [b]$, let $S_i \subseteq [M]$ be a set such that $\abs{S_i}=N$ and $S_i \subseteq I_i$, and $f_i: [N]\to S_i$ be a bijection. Let us fix the permutation $\pi: [M]\to [M]$, defined as,
    \begin{align}
        \pi(x) = \begin{cases}
            f_{i+1}(j) &\text{if $x = f_i(j)$ for $j\in [N], i\in [b-1]$}\\
            j &\text{if $x = f_b(j)$ for $j\in [N]$}\\
            x &\text{if $x\notin \bigcup_i S_i$}.
            \label{eq:permutation_def}
        \end{cases}
    \end{align}
    The problem is to decide whether (YES) $\abs{S_b^{\textup{even}}} = N$ or (NO) $\abs{S_b^{\textup{even}}} \leq \tfrac{N}{2}$.
\end{problem}

Let us show that the problem is in \QMAo.

\begin{restatable}[\QMAo algorithm]{lemma}{lemQMAoAlgorithm}
    For any in-place oracle $\Pi=\{\pi_n\}_n$ where $\pi_n$ are instances of \cref{prob:perm_pointer_chasing}, the language $L\coloneq \{1^n: \pi_n \text{ is a YES instance}\}$ is in $\QMAo^{\Pi}[1,1]$.
\end{restatable}
\begin{proof}
    The proof follows the proof of Theorem 4 in~\cite{FK18}, except we succeed with perfect completeness. For the sake of self-containment, we write the proof below. The verifier algorithm $V$ is as follows.
    \begin{algobox}{\QMAo verifier $V$}
    \Input{Oracle access to $O_{\pi_n}$ and $ln$-bit state $\ket{\psi}$.}
        Each with probability $\tfrac{1}{2}$, run one of the following subroutines,
        \begin{enumerate}[label=(\alph*)]
            \item Apply the oracle $O_{\pi_n}$ to $\ket{\psi}$, followed by $I^{\otimes (l-1)n} \Had^{\otimes n}$. Accept if and only if the measurement output is $0^{ln}$.
            \item Measure $\ket{\psi}$ in the standard basis, obtaining output $o$. If $o\in [M]^{\textup{odd}}$, reject. Otherwise, apply $O_{\pi_n}$ to $o$ and accept if and only if $O_{\pi_n}\ket{o} \in [N]$.
        \end{enumerate}
    \end{algobox}

    Let us argue for completeness. In the YES case, the verifier may receive $\ket{\psi}= \pi_n^{-1}([N])$. In this case, subroutine (a) accepts with probability $1$ as $O_{\pi_n}\ket{\psi} = \ket{[N]}$ and $(I^{\otimes (l-1)n}\otimes \Had^{\otimes n})\ket{0^n} = \ket{[N]}$. Similarly, subroutine (b) also accepts with probability $1$.

    Next, we show soundness. Consider a NO instance and assume  $\ket{\psi} = \sum_{o\in [M]} \alpha_o \ket{o}$ is the witness. Letting $p_{(a)}, p_{(b)}$ be the respective acceptance probabilities of each subroutine,
    \begin{align}
        p_{(a)} &= \abs{\bra{[N]} O_{\pi_n} \ket{\psi}}^2 = \frac{1}{N} \abs{\sum_{o \in S_b} \alpha_o}^2\label{eq:p_a}\\
        p_{(b)} &= \sum_{o\in S_b^{\textup{even}}} \abs{\alpha_o}^2.\label{eq:p_b}
    \end{align}
    Letting $s \coloneq \tfrac{\abs{S_b^{\textup{even}}}}{N} \leq \tfrac{1}{2}$, by the triangle inequality and Cauchy-Schwarz,
    \begin{align*}
        \abs{\sum_{o\in S_b} \alpha_o} &\leq \abs{\sum_{o\in S_b^{\textup{even}}} \alpha_o} + \abs{\sum_{o\in S_b^{\textup{odd}}} \alpha_o}\\
        &\leq \sqrt{s p_{(b)}} + \sqrt{(1-s)(1 - p_{(b)})}.
    \end{align*}
    Using \cref{eq:p_a},
    \begin{align*}
        p_{(a)} \leq \left(\sqrt{s p_{(b)}} + \sqrt{(1-s)(1-p_{(b)})} \right)^2.
    \end{align*}
    Therefore, the verifiers acceptance probability is,
    \begin{align*}
        \frac{p_{(a)} + p_{(b)}}{2} &\leq \frac{1}{2} \left( \left(\sqrt{sp_{(b)}} + \sqrt{(1-s)(1-p_{(b)})} \right)^2 + p_{(b)} \right)\\
        &=\frac{1-s}{2} + sp_{(b)} + \sqrt{sp_{(b)}(1-s)(1-p_{(b)})} \leq \frac{1 + \sqrt{s}}{2} < 0.9,
    \end{align*}
    where the second inequality in the last line is due to the fact that the equation is maximized at $p_{(b)} = \frac{1 + \sqrt{s}}{2}$. Thus, the completeness and soundness gap is at least $0.1$, which can be amplified by standard techniques~\cite{MW05}.
\end{proof}

Next, we show a \QCMA lower bound on the problem.

\begin{restatable}[\QCMA lower bound]{lemma}{lemQCMAInplaceLowerbound}\label{lem:qcma_inplace_lowerbound}
    Fix $l\geq 4$, $\alpha\in (0,\tfrac{1}{2})$, and $\beta\geq 0$ such that,
    \begin{align}
        \beta + 5\alpha < l - 2.\label{eq:beta_alpha_condition}
    \end{align}
    There exists an in-place oracle $\Pi = \{\pi_n \}_n$ such that the language $L\coloneq \{1^n: \pi_n \text{ is a YES instance}\}$ is not in $\QCMA^{\Pi}[b(n)-2, q(n)]$ where $q(n) \leq 2^{\beta n}$.
\end{restatable}

\begin{proof}
    Let us fix some instance $n$. We first fix an in-place algorithm $V$ which takes in some witness $v\in \{0,1\}^{w(n)}$ and makes $R\coloneq b-2$ rounds of $q$ parallel queries per round. Let $\ket{\psi_0} = \ket{0}\ket{v}$ be the initial state and for any permutation $\sigma$ and $r\in [R]$, let $\ket{\psi_r^\sigma}$ denote the state right before the $r+1$ query,
    \begin{align*}
        \ket{\psi_r^\sigma} \coloneq U_r O_{\sigma}^{\otimes q} \dots O_{\sigma}^{\otimes q} U_0 \ket{\psi_0}.
    \end{align*}
    
    The distribution $D_n^{\textup{YES}}$ over YES instances of $\pc_n$ is defined as follows. First, set $S_1 = [N]$ and $f_1(j) = j$. For $i\in [2,b-1]$, draw $S_i$ uniformly from $\binom{I_i}{N}$ and $f_i$ be a uniformly-random bijection. Last, $S_b$ is uniformly-drawn from $\binom{I_{b}^{\textup{even}}}{N}$ and $f_b$ is uniformly-random. Given a fixed element from $D_n^{\textup{YES}}$, we may define $\pi_n$ using \cref{eq:permutation_def}.

    Next, we define the \emph{transcript} and its \emph{canonical completion}. For each round $r\in [R]$, the transcript $\tsc[r]$ tracks the known and unknown parts of each set $S_i$ of an oracle. It consists of the following parts:
    \begin{enumerate}[label=(\alph*)]
        \item \emph{Discovered} sets $D_i^{(r)} \subseteq S_i$,
        \item \emph{Commited} sets $Q_i^{(r)} \subseteq S_i$ and their associated commited values $\pi(x)$ for all $x\in Q_i^{(r)}$,
        \item \emph{Frontier} sets $F_i^{(r)} \coloneq D_i^{(r)} \setminus Q_i^{(r)}$,
        \item \emph{Unexplored} sets $U_i^{(r)} \subseteq I_i \setminus D_i^{(r)}$,
        \item \emph{Excluded} sets $E_i^{(r)} \coloneq I_i \setminus (U_i^{(r)} \cup D_i^{(r)})$.
    \end{enumerate}
    We initialize $\tsc[0]$ as follows,
    \begin{align*}
        &D_1^{(0)} = F_1^{(0)} = [N]  &U_1^{(0)}=Q_1^{(0)} = E_1^{(0)} = \emptyset,\\
        &\forall i\geq 2,\; U_i^{(0)} = I_i &D_i^{(0)}=Q_i^{(0)} = E_i^{(0)} = F_i^{(r)} = \emptyset.
    \end{align*}
    Notice that we may use $\tsc[r]$ to partition $I_i$ as follows,
    \begin{align}
        I_i = Q_i^{(r)} \sqcup F_i^{(r)} \sqcup U_i^{(r)} \sqcup E_i^{(r)}.\label{eq:partition_i}
    \end{align}
    Assuming $\abs{\bigcup_i F_i^{(r)}} = \abs{[N] \setminus \pi(Q_b^{(r)})}$, we define the root function $f^{(r)}(x): \bigcup_i F_i^{(r)} \to [N] \setminus \pi(Q_b^{(r)})$ to be the first bijection from an enumerated list of functions. The \emph{canonical completion} $\tilde{\pi}_r$ is defined as,
    \begin{align}
        \tilde{\pi}_r(x) \coloneq \begin{cases}
            \pi(x) & \text{if } x \in \bigcup_i Q_i^{(r)}, \\
            f^{(r)}(x) & \text{if } x \in \bigcup_i F_i^{(r)}, \\
            x & \text{if } x \in I_{\textup{junk}} \cup \bigcup_i E_i^{(r)} \cup \bigcup_i U_i^{(r)}.
        \end{cases}\label{eq:completion_def}
    \end{align}
    Intuitively, each point $x \in \bigcup_i F_i^{(r)}$ is at the end of a known chain from $[N]$. Using $f^{(r)}$, we ensure that $\tilde{\pi}_r(x)$ is a permutation.
    
    \begin{mathinlay}
        \begin{claim}\label{clm:canonical_permutation}
            Assume that for a $\tsc[r]$ based on permutation $\pi$ and for all $x\in \bigcup_i F_i^{(r)}$, there exists some sequence of points $x_0,\dots x_{i-1}, x$ such that for all $t\in [i-1]$, $x_t\in Q_{t}^{(r)}$. Then,
            \begin{enumerate}
                \item there exists a bijection $f^{(r)}(x): \bigcup_i F_i^{(r)} \to [N] \setminus \pi(Q_b^{(r)})$,\label{itm:bijection_existence}
                \item $\tilde{\pi}_r$ is a permutation and,\label{itm:canonical_permutation}
                \item for all $x\in [M]$, $\pi(x)=x \implies \tilde{\pi}_r(x) = x$.\label{itm:self_loops_canonical}
            \end{enumerate}
        \end{claim}
        \begin{proof}
            Based on \cref{eq:completion_def}, partition $[M] = A\sqcup B \sqcup C$ where $A \coloneq \bigcup_i Q_i^{(r)}$, $B \coloneq \bigcup_i F_i^{(r)}$, and $C \coloneq I_{\textup{junk}} \cup \bigcup_i E_i^{(r)} \cup \bigcup_i U_i^{(r)}$.

            Due to the assumption on $x\in \bigcup_i F_i^{(r)}$, we have that $x\in \bigcup_i F_i^{(r)}$ if and only if for all $z\in [b-i]$, $\pi^z(x) \in U_i^{(i+z)}$.\footnote{By $\pi^z(x)$, we mean applying the permutation on input $x$ $z$-times.} As the number of such points is $\abs{[N] \setminus \pi(Q_b^{(r)})}$, \cref{itm:bijection_existence} follows.

            Let us show \cref{itm:canonical_permutation}. As elements in $C$ are mapped to themselves in $\tilde{\pi}_r$, it suffices to show that $\tilde{\pi}_r|_{A\cup B}$ is a bijection on $A\cup B$. By definition of $A$, the only points $x\in A$ for which there does not exist a $y\in A$ such that $x = \tilde{\pi}_r(y)\in A$ are in $[N]\setminus \pi(Q_b^{(r)})$. Similarly, all points in $x\in A$ such that $\tilde{\pi}_r(x)\notin A$ are in $\bigcup_i F_i^{(r)}$. Per \cref{itm:bijection_existence} combined with the fact that $\pi$ is a permutation, $\tilde{\pi}_r|_{A\cup B}$ is a permutation.

            By \cref{eq:permutation_def}, $\pi(x)=x$ implies that $x\notin \bigcup_i S_i$, so $x\notin \bigcup_i D_i^{(r)}=A\cup B$. Therefore, $x\in C$, proving \cref{itm:self_loops_canonical}.
        \end{proof}
    \end{mathinlay}

    The canonical state $\ket{\psi_r^\prime}$ is defined to be,
    \begin{align*}
        \ket{\psi_r^\prime} = U_r O_{\tilde{\pi}_r}^{\otimes q} U_{r-1} O_{\tilde{\pi}_{r-1}}^{\otimes q} \dots O_{\tilde{\pi}_1}^{\otimes q} U_0 \ket{\psi_0},
    \end{align*}

    which depends on the transcripts $\tsc[i]$ for $i\in[r]$. 
    We say that a permutation $\sigma$ is \emph{consistent} with $\tsc[r]$ if $\tsc[r]$ is its transcript after $r$ rounds under a verifier $V$. Furthermore, we define heavy sets, the sets which have a large weight in the query registers of $\ket{\psi_r^\pi}$ (i.e. right before the oracle call $O_\pi^{\otimes q}$). Formally, let $r\in [R]$, $i \in [b]$ and $j\in [q]$ be some variables. Letting $X_j$ denote the $j$-th query register and $\Pi_S^{(j)} \coloneq \mathbf{I}^{\otimes(j-1)}\otimes\bigl(\sum_{x\in S}\ketbra{x}{x}\bigr)\otimes\mathbf{I}^{\otimes(q-j)}$ the projector onto $S\subseteq[M]$ acting on the $j$-th query register, consider the following variables,
    \begin{align*}
        p_{i,j}^{(r)} &\coloneq \bra{\psi_{r-1}^\prime}\Pi_{I_i}^{(j)}\ket{\psi_{r-1}^\prime}\\
        p_i^{(r)} &\coloneq \sum_{j\in [q]} p_{i,j}^{(r)}\\
        \mu_i^{(r)}(S) &\coloneq \frac{1}{p_{i}^{(r)}} \sum_{j\in [q]} \bra{\psi^\prime_{r-1}}\Pi_S^{(j)}\ket{\psi^\prime_{r-1}} \qquad (S\subseteq I_i).
    \end{align*}
    Note that when $p_i^{(r)} = 0$, we let $\mu_i^{(r)}(S) = 0$. The \emph{heavy} set $H_i^{(r)}$ is defined as,
    \begin{align*}
        H_i^{(r)} \coloneq \left\{ x\in I_i: \mu_i^{(r)}(\{x\}) \geq \frac{\epsilon}{N} \right\},
    \end{align*}
    where $\epsilon$ is a variable we will set later. Next, let us define how the transcript is \emph{updated} each round. For $r\geq 1$, let the undiscovered part of $S_i$ be $S_i^{(r)} \coloneq S_i \cap U_i^{(r-1)}$ and the heavy hits be $G_i^{(r)} \coloneq H_i^{(r)} \cap S_i^{(r)}$.

    For each $g\in G_i^{(r)}$, let $\mathrm{chain}(g)$ denote the following maximal sequence,
    \begin{align*}
        \mathrm{chain}(g) = g, g^{-1}, \dots ,g^{-k}: \forall m\in[k-1]: \pi^{-1}(g^{-m}) = g^{-(m+1)} \text{ and } g^{-k} \in S_1.
    \end{align*}
    Essentially, $\mathrm{chain}(g)$ represents the sequence of steps of $\pi$ which bring us from $g$ back to $S_1$. The \emph{path points} at level $i$ are elements of $\mathrm{chain}(g)$ lying in $I_i$ that have not yet been committed,
    \begin{align*}
        P_i^{(r)} \coloneq \left\{ z: \exists\, k> i, \exists\, g\in G_k^{(r)} \text{ s.t. }z\in \mathrm{chain}(g)\cap I_i  \right\}.
    \end{align*}

    We commit these path points to the transcript in order to ensure that we may define a valid root function $f^{(r)}$ and completion $\tilde{\pi}_{r}$ in \cref{eq:completion_def}.
    For $i\geq 2$ and going from round $r-1$ to $r$, the transcript is updated as follows,
    \begin{align}
        Q_i^{(r)} &\coloneq Q_i^{(r-1)} \cup D_i^{(r-1)} \cup G_i^{(r)} \cup P_i^{(r)}\label{eq:q_i_update}\\
        D_i^{(r)} &\coloneq D_i^{(r-1)} \cup G_i^{(r)} \cup P_i^{(r)} \cup \pi\!\bigl(Q_{i-1}^{(r)}\bigr)\label{eq:d_i_update}\\
        U_i^{(r)} &\coloneq U_i^{(r-1)} \setminus \bigl(H_i^{(r)} \cup D_i^{(r)}\bigr).\label{eq:u_i_update}
    \end{align}
    Note that for $i=1$ and $r>0$, the committed set is $Q_1^{(r)} = [N]$.
    One may see that the condition in \cref{clm:canonical_permutation} is satisfied for all $r$.
    In summary, the order of operations is as follows. Given some canonical state $\ket{\psi_{r-1}^\prime}$, we obtain the update probabilities. In turn, this updates the transcript $\tsc[r]$, whose commited values in $Q_i^{(r)}$ are based on the original permutation $\pi$. In turn, we obtain the canonical completion $\tilde{\pi}_r$ which leads to the state $\ket{\psi_{r}^\prime}$.

    Next, we will show that two oracles with the same transcript up to the step $r$ are difficult to distinguish.

    \begin{mathinlay}
        \begin{claim}[Transcript indistinguishability]\label{clm:transcript_indistinguishability}
            For any pair of oracles $\pi_1, \pi_2$ consistent with the same transcript history $\tsc[t]$ for $t\in[r]$,
            \begin{align*}
                \norm{\ket{\psi_r^{\pi_1}} - \ket{\psi_r^{\pi_2}}} \leq 4r\sqrt{q\epsilon}.
            \end{align*}
        \end{claim}
        \begin{proof}
            To show the statement, it suffices to show that for any permutation $\sigma$ consistent with $\tsc[r]$,
            \begin{align}
                \norm{\ket{\psi_r^{\sigma}} - \ket{\psi_r^{\prime}}} \leq 2r\sqrt{q\epsilon},\label{eq:transcript_indist_final}
            \end{align}
            as applying \cref{eq:transcript_indist_final} to $\pi_1$ and $\pi_2$ and using the triangle inequality gives the claim. We prove \cref{eq:transcript_indist_final} by induction on $t\in [0, r]$. The base case $t=0$ is immediate since the witness and verifier are fixed. For the inductive step,
            \begin{align}
                \norm{\ket{\psi_t^{\sigma}} - \ket{\psi_t^{\prime}}} 
                &= \norm{O_{\sigma}^{\otimes q} \ket{\psi_{t-1}^{\sigma}} - O_{\tilde{\pi}_t}^{\otimes q}\ket{\psi_{t-1}^{\prime}}}\nonumber\\
                &\leq \norm{\ket{\psi_{t-1}^{\sigma}} - \ket{\psi_{t-1}^{\prime}}} + \norm{(O_{\sigma}^{\otimes q} - O_{\tilde{\pi}_t}^{\otimes q}) \ket{\psi_{t-1}^{\prime}}}\nonumber\\
                &\leq 2(t-1)\sqrt{q\epsilon} + \norm{(O_{\sigma}^{\otimes q} - O_{\tilde{\pi}_t}^{\otimes q}) \ket{\psi_{t-1}^{\prime}}},\label{eq:triangle_split}
            \end{align}
            where the equality follows due to unitarity of $U_t$ and the inequality is an application of the triangle inequality and inserting $O_\sigma^{\otimes q} \ket{\psi_{t-1}^\prime}$.

            We claim that $\sigma$ and $\tilde{\pi}_t$ may disagree only on $\bigcup_{i}(S_i^{(t)}\setminus H_i^{(t)})$. This can be shown case by case. First, if $x\in I_{\textup{junk}}$, both oracles are the identity. Let us check each part of the partition of $I_i$ in \cref{eq:partition_i},
            \begin{enumerate}
                \item If $x\in Q_i^{(t)}$, then by \cref{eq:completion_def} both $\sigma$ and $\tilde{\pi}_t$ agree with $\pi(x)$.
                \item If $x\in F_i^{(t)} = D_i^{(t)}\setminus Q_i^{(t)}$, then $x\in S_i^{(t)}\setminus H_i^{(t)}$. This is due to the fact that since $D_i^{(t)}\setminus Q_i^{(t)} \subseteq \pi(Q_{i-1}^{(t)})\setminus Q_i^{(t)} \subseteq S_i\setminus D_i^{(t-1)}$, we have $x\in U_i^{(t-1)}\cap S_i = S_i^{(t)}$. Furthermore, $x\notin H_i^{(t)}$ as otherwise $x\in G_i^{(t)}\subseteq Q_i^{(t)}$, contradicting $x\in F_i^{(t)}$.
                \item If $x\in U_i^{(t)}$, then $\sigma$ and $\tilde{\pi}_t$ may only disagree on $U_i^{(t)} \cap S_i$ as all other points are self-loops. As $U_i^{(t)} \subseteq U_i^{(t-1)}$, $D_i^{(t)} \cap U_i^{(t)} = \emptyset$ and $G_i^{(t)} \subseteq D_i^{(t)}$, we have $U_i^{(t)} \cap S_i \subseteq S_i^{(t)} \setminus H_i^{(t)}$ (the $H_i^{(t)}$ elements in $U_i^{(t-1)}$ were added to $G_i^{(t)}\subseteq Q_i^{(t)}\subseteq D_i^{(t)}$, so they left $U_i^{(t)}$).
                \item If $x\in E_i^{(t)}$, there is no disagreement as $\pi(x)=x =\tilde{\pi}_r(x)$.
            \end{enumerate}
            In all cases, potential disagreement is contained in $S_i^{(t)} \setminus H_i^{(t)}$. Letting $p_j$ be the probability that the $j$-th query register holds a value in $\bigcup_{i\in [b]} (S_i^{(t)} \setminus H_i^{(t)})$ and using the fact that each $x \notin H_i^{(t)}$ has $\mu_i^{(t)}(\{x\}) < \epsilon/N$,
            \begin{align*}
                \mu_i^{(t)}(S_i^{(t)} \setminus H_i^{(t)}) \leq \abs{S_i^{(t)}} \cdot \frac{\epsilon}{N} \leq \epsilon,
            \end{align*}
            we have that,
            \begin{align*}
                \sum_{j\in [q]} p_j &= \sum_{i\in [b]} p_i^{(t)} \mu_i^{(t)}(S_i^{(t)} \setminus H_i^{(t)})\\
                &\leq \epsilon \sum_{i\in [b]} p_i^{(t)} \leq \epsilon q.
            \end{align*}
            By \cref{lem:parallel_hybrid}, the second term of \cref{eq:triangle_split} is bounded by $2\sqrt{q\epsilon}$, finishing the proof.
        \end{proof}
    \end{mathinlay}

    Let $\delta \in (0,1]$ be a constant we fix later. We say a transcript $\tsc[r]$ is \emph{good} if,
    \begin{align*}
        \forall i\in [r+2, b]: \abs*{D_i^{(r)}} \leq 4rb\delta N\\
        \forall i\in [r+1, b]: \abs*{Q_i^{(r)}} \leq 4rb\delta N.
    \end{align*}
    If $\tsc[r]$ does not satisfy the criteria above, we call it \emph{bad}.
    \begin{mathinlay}
        \begin{claim}[Probability of a bad transcript]\label{clm:bad_transcript}
            Assume $4rb\delta \leq 1$ and $\left(\tfrac{1}{\epsilon} + 4b\delta\right) rN \leq \tfrac{B}{4}$. Then for a fixed verifier $V$ and witness $v$,
            \begin{align}
                \Pr_{\pi \sim D_n^{\textup{YES}}} \left[ \tsc[r] \text{ is bad} \right] \leq rb \left( \frac{4eN}{\delta \epsilon B} \right)^{\delta N}.\label{eq:bad_transcrip_probability}
            \end{align}
        \end{claim}
        \begin{proof}
            We prove this by induction on $r \in [R]$. The base case holds as, except for $D_1^{(0)}$, $D_i^{(0)} = Q_i^{(0)} = \emptyset$. For the inductive step, assume that $\tsc[r-1]$ is good and fix some $i\in [r+1,b]$. As $i \geq r+1 \geq (r-1)+2$, goodness of $\tsc[r-1]$ and the condition $\left(\tfrac{1}{\epsilon}+4b\delta\right)(r-1)N \leq B/4$ imply $\abs{U_i^{(r-1)}} \geq \tfrac{3B}{4}$.

            If $i\leq b-1$, then $S_i^{(r)} = S_i \cap U_i^{(r-1)}$ is a uniformly-random subset of $U_i^{(r-1)}$ of size $N-\abs{D_i^{(r-1)}} \leq N$. For $i=b$, $S_b^{(r)}$ is uniform over $U_{b}^{(r-1)} \cap [M]^{\textup{even}}$; since $I_b$ has $\geq B/2$ even elements and at most $B/4$ total have been removed, $\abs{U_{b}^{(r-1)} \cap [M]^{\textup{even}}} \geq \frac{B}{4}$. In both cases, as $H_i^{(r)}$ is determined by $\tsc[r-1]$ (as the verifier and witness are fixed), $S_i^{(r)}$ is stochastically dominated by a uniformly-chosen set of size $N$ from a universe of size at least $\frac{B}{4}$. By \cref{lem:too_heavy},
            \begin{align}
                \Pr \left[ \abs{G_i^{(r)}} \geq \delta N \mid \tsc[r-1] \right] \leq \left( \frac{4eN}{\delta \epsilon B} \right)^{\delta N}.\label{eq:heavy_hit_bound}
            \end{align}

            Let $E$ denote the event that for all $t\leq r$ and $i\geq t+1$, $\abs{G_i^{(t)}}<\delta N$.
            Next, we show that if $E$ occurs, then $\tsc[r]$ is good. Let,
            \begin{align*}
                J^{(r)} \coloneq \{j\in [N]: \exists t \leq r, k > t \text{ s.t. } f_k(j)\in G_k^{(t)} \},
            \end{align*}
            denote the indices of \emph{paths} which were guessed earlier than expected. Under the assumption on $\abs{G_i^{(t)}}$, we have that $\abs{J^{(r)}} \leq rb\delta N$. We claim that every point in $D_i^{(r)}$ with $i\geq r+2$ and $Q_i^{(r)}$ with $i\geq r+1$ has a label in $J^{(r)}$. We prove this by induction on $r$. The claim is immediate for $r=0$. For the inductive step, consider $Q_i^{(r)}$ with $i\geq r+1$. By \cref{eq:q_i_update},
            \begin{align*}
                Q_i^{(r)} \subseteq Q_i^{(r-1)}\cup D_i^{(r-1)}
                \cup G_i^{(r)}\cup P_i^{(r)}.
            \end{align*}
            By the inductive hypothesis, the terms $Q_i^{(r-1)}$ and $D_i^{(r-1)}$ have labels in $J^{(r-1)} \subseteq J^{(r)}$. The terms $G_i^{(r)}$ have labels in $J^{(r)}$ as $i\geq r+1$, while a label $j$ in $P_i^{(r)}$ implies that it is in some $G_k^{(r)}$ where $k > i$. By similar reasoning, we find the same for $D_i^{(r)}$. Therefore $\abs{J^{(r)}}$ bounds the size of all $D_i^{(r)}$ and $Q_i^{(r)}$ we care about, meaning $\tsc[r]$ is good. 
            
            All that remains is bounding the probability of event $E$. Applying \cref{eq:heavy_hit_bound}, union-bounding over $b$ buckets with $i \geq r+1$ and applying the chain rule over $r$ rounds,
            \begin{align*}
                \Pr_{\pi \sim D_n^{\textup{YES}}} \left[ \tsc[r] \text{ is bad} \right] \leq rb \left( \frac{4eN}{\delta \epsilon B} \right)^{\delta N}.
            \end{align*}
        \end{proof}
    \end{mathinlay}
    Lastly, let us show that for any YES instance, we may create a NO instance whose state is extremely close to that of the YES instance.
    \begin{mathinlay}
        \begin{claim}[NO-instance transformation]\label{clm:no_instance}
            Suppose that for all $r\in[R]$, $\tsc[r]$ is good for $\pi\in \supp(D_n^{\textup{YES}})$ where $R=b-2$, $8Rb\delta \leq \tfrac{1}{4}$ and $(\tfrac{1}{\epsilon} + 4b \delta)RN\leq \tfrac{B}{4}$. Then there exists a NO instance $\pi_{\textup{NO}}$ of $\pc_n$ which is consistent with $\tsc[r]$.
        \end{claim}
        \begin{proof}
            We may entirely describe the permutation $\pi$ using $f_1,\dots f_b$. Notice that by definition, $B$ is always even. As $\pi\in \supp(D_n^{\textup{YES}})$, $f_b([N]) = S_b \subseteq I_b^{\textup{even}}$. Let $J_{\textup{fix}}$ denote the fixed points which certify that some $x\in S_b$ is even,
            \begin{align*}
                J_{\textup{fix}} \coloneq \left\{j: f_{b-1}(j)\in Q_{b-1}^{(R)} \right\} \cup \left\{ j: f_b(j)\in D_b^{(R)} \right\}.
            \end{align*}
            As $\tsc[R]$ is good, $\abs{J_{\textup{fix}}} \leq 8Rb\delta N \leq \tfrac{N}{4}$. Let $F \coloneq \{f_b(j): j\in J_{\textup{fix}}\} \subseteq I_b^{\textup{even}}$.
            Per our assumption, elements leave $U_b$ by entering $D_b$ or $H_b$. Over $R$ rounds, at most $\abs{D_b^{(R)}} + \sum_{r\in[R]}\abs{H_b^{(r)}} \leq 4Rb\delta N + RN/\epsilon \leq \left(\frac{1}{\epsilon}+4b\delta\right)RN \leq \frac{B}{4}$ elements have been removed, so $\abs{U_b^{(R)}} \geq \frac{3B}{4}$. Since $B$ is even, $I_b$ is an interval of size $B$ with exactly $B/2$ even elements and $B/2$ odd elements. At most $B/4$ elements in total have been removed, so at most $B/4$ even (resp.\ odd) elements were removed. Hence $\abs{U_b^{(R)}\cap I_b^{\textup{even}}} \geq \frac{B}{2} - \frac{B}{4} = \frac{B}{4}$ and similarly $\abs{U_b^{(R)} \cap I_b^{\textup{odd}}} \geq \frac{B}{4}$. Note also $F \cap U_b^{(R)} = \emptyset$ since $F \subseteq D_b^{(R)}$.

            Choose $O^\prime \subseteq U_b^{(R)} \cap I_b^{\textup{odd}}$ and $E^\prime \subseteq U_b^{(R)} \cap I_b^{\textup{even}}$ such that $\abs{O^\prime} = \tfrac{N}{2}$ and $\abs{E^\prime} = \tfrac{N}{2} - \abs{F}$ (feasible since both sets have size $\geq B/4 \geq N$). Then we set $S_b^{\textup{NO}} \coloneq F \cup E^\prime \cup O^\prime$, giving $\abs{S_b^{\textup{NO}}} = N$ and $\abs{(S_b^{\textup{NO}})^{\textup{odd}}} = N/2$.

            Define $f_b^{\textup{NO}}: [N]\to S_b^{\textup{NO}}$ by $f_b^{\textup{NO}}(j) = f_b(j)$ for all $j\in J_{\textup{fix}}$ and bijectively assigning the remaining indices to $S_b^{\textup{NO}}\setminus F$ (feasible since $\abs{S_b^{\textup{NO}}\setminus F}=N-\abs{J_{\textup{fix}}}$). We define $\pi_{\textup{NO}}$ using $f_1,\dots,f_{b-1}$ unchanged and replacing $f_b$ with $f_b^{\textup{NO}}$.

            We claim that $\pi_{\textup{NO}}$ generates the identical transcript to $\pi$ through every round $r\in[R]$. Since $f_1,\dots,f_{b-1}$ are unchanged, the heavy-hit sets $G_i^{(r)}$ and path points $P_i^{(r)}$ for buckets $i\leq b-1$ are identical for $\pi$ and $\pi_{\textup{NO}}$ at every round. For bucket $b$: every point of $U_b^{(R)}$ avoided $D_b^{(r)}$ and $H_b^{(r)}$ for all $r\leq R$ (otherwise it would have left $U_b$ before round $R$). Since $E^\prime, O^\prime\subseteq U_b^{(R)}$, replacing elements of $S_b$ inside $U_b^{(R)}$ by other elements of $U_b^{(R)}$ cannot change any $G_b^{(r)}$ or $D_b^{(r)}$ for $r\leq R$. Moreover, for $j\in J_{\textup{fix}}$, we have $f_b^{\textup{NO}}(j) = f_b(j)$, preserving all committed values. It follows inductively on $r$ that the entire transcript $\tsc[r]$ is identical for $\pi$ and $\pi_{\textup{NO}}$ for every $r\in[R]$.
        \end{proof}
    \end{mathinlay}
    With the claims above, we may finally set the parameters $\epsilon$ and $\delta$,
    \begin{align*}
        \epsilon\coloneq\frac{1}{10^4 q b^2}, \;\; \delta \coloneq \frac{1}{64b^2}.
    \end{align*}
    We may see that with these assignments, the conditions of \cref{clm:bad_transcript,clm:no_instance} are satisfied. The first condition holds as,
    \begin{align*}
        8Rb\delta \leq 8b^2 \cdot \frac{1}{64b^2} \leq \frac{1}{4}.
    \end{align*}
    The second condition is due to the following
    \begin{align*}
        \left(\frac{1}{\epsilon} + 4b\delta \right)RN \leq 2 \cdot 10^4 qb^3 N,
    \end{align*}
    which means we are showing that $2 \cdot 10^4 qb^3 N \leq \tfrac{B}{4}$. For large $n$, $B \geq \tfrac{M-N}{b-1} \geq \tfrac{M}{2b}$, so it suffices to show $16\cdot 10^4\, qb^4 N \leq M$. Substituting $q\leq 2^{\beta n}$, $b\leq 2^{\alpha n+1}$, and $N=2^n$,
    \begin{align*}
        16\cdot 10^4 qb^4 N \leq 256\cdot 10^4 \cdot 2^{(\beta + 4\alpha +1)n}.
    \end{align*}
    As $M=2^{ln}$, this holds whenever $256\cdot 10^4\cdot 2^{(\beta+4\alpha+1)n} \leq 2^{ln}$, meaning for all $n$ such that,
    \begin{align}
        n \geq n_1 \coloneq \left\lceil \frac{\log_2(256\cdot 10^4)}{l - 1 - \beta - 4\alpha} \right\rceil,\label{eq:n1_threshold}
    \end{align}
    which is finite since $l-1-\beta-4\alpha > 0$ by \cref{eq:beta_alpha_condition}. Therefore, for all $n\geq n_1$, the conditions of the claims are satisfied. Furthermore, these variables mean that $4R\sqrt{q\epsilon}$ is always bounded by $\tfrac{1}{25}$. Lastly, we bound the bad transcript probability in \cref{eq:bad_transcrip_probability}. The base is bounded by,
    \begin{align*}
        \frac{4eN}{\delta \epsilon B} &\leq 512e \cdot 10^4 \cdot \frac{qb^5 N}{M}\\
        &\leq 512 \cdot 32 \cdot e \cdot 10^4 \cdot 2^{(\beta + 5\alpha +1 - l)n} \leq C \cdot 2^{-\lambda n},
    \end{align*}
    where $C = 16384\cdot e \cdot 10^4$ is a constant and $\lambda = l - 1-\beta -5\alpha$. By \cref{eq:beta_alpha_condition}, $\lambda > 1$. Therefore for $n \geq n_2 = \lceil(\log_2 C + 1)/\lambda\rceil$, the base is bounded by $\tfrac{1}{2}$. Therefore,
    \begin{align}
        \Pr_{\pi \sim D_n^{\textup{YES}}} \left[ \tsc[r] \text{ is bad} \right] \leq rb \cdot 2^{-\delta N}.\label{eq:bad_transcript_bound}
    \end{align}

    Let us show argue for the \QCMA hardness. Fix any \QCMA verifier $V$ with $R=b-2$ rounds of $q$ queries and a witness of size $w=\poly(n)$. By taking the union bound over \cref{eq:bad_transcript_bound}, we have that,
    \begin{align*}
        \Pr_{\pi \sim D_n^{\textup{YES}}}[\exists v\in \{0,1\}^w: \tsc[r]_v \text{ is bad}] \leq 2^w \cdot Rb\cdot 2^{-\delta N},
    \end{align*}
    where $\tsc[r]_v$ denotes the transcript using verifier $V$ and witness $v$. Let $w \leq c\cdot n^d$ for some constants $c,d$ and notice that,
    \begin{align*}
        \delta N = \frac{N}{64b^2} \geq \frac{2^{(1-2\alpha)n}}{256}.
    \end{align*}
    Therefore,
    \begin{align}
        2^w \cdot Rb \cdot 2^{-\delta N} \leq 2^{2 + 2\alpha n + cn^d} \cdot 2^{-\frac{2^{(1-2\alpha)n}}{256}}.\label{eq:final_bound}
    \end{align}
    We have that this quantity is strictly less than $1$ for all $n$ such that,
    \begin{align*}
        n > n_3 \coloneq \frac{256(c+2\alpha + 2) (d+1)! }{(1-2\alpha)^{d+1}(\ln 2)^{d+1}}.
    \end{align*}
    Therefore, for all $n \geq \max(n_1, n_2, n_3)$, there exists a yes-instance $\pi^*_{\textup{YES}} \in \supp(D_n^{\textup{YES}})$ such that $\tsc[R]_v$ is good for every witness $v$. However, by \cref{clm:transcript_indistinguishability,clm:no_instance}, $V$ cannot distinguish $\pi^*_{\textup{YES}}$ from a no-instance $\pi^*_{\textup{NO}}$ consistent with $\tsc[R]$ of $\pi^*_{\textup{YES}}$ with probability above $\tfrac{1}{25}$. Therefore, the completeness condition is not satisfied for $\pi^*_{\textup{YES}}$, or soundness for $\pi^*_{\textup{NO}}$. Letting $\pi^*_n$ being one of these, by standard diagonalization techniques, we may create the oracle family $\Pi = \{\pi^*_n\}$ such that $L\notin \QCMA^{\Pi}[b(n)-2, q(n)]$.
\end{proof}

Combining the two results above, we obtain the following theorem.

\begin{theorem}[In-place oracle separation between \QCMA and \QMAo]\label{thm:inplace-qma-qcma}
    There exists an in-place oracle $\Pi$ with respect to which, $\QMAo^{\Pi} \not\subseteq \QCMA^\Pi$.
\end{theorem}

\subsection{Bounded adaptivity standard oracle separation}

One may observe that the proof of \cref{lem:qcma_inplace_lowerbound} works using classical oracles.

\begin{remark}\label{rem:qcma_classical_lowerbound}
    Define the oracle family $O = \{O_n\}$ as the family of instances of \cref{prob:perm_pointer_chasing} $\pi_n$ where,
    \begin{align}
        O_n \ket{x,y} \coloneq \ket{x, y\oplus \pi_n(x)}\label{eq:classical_oracle_permutation}
    \end{align}
    There exists an oracle $O$ such that the language $L\coloneq \{1^n: \pi_n \text{ is a YES instance}\}$ is not in $\QCMA^{O}[b(n)-2,q(n)]$, where $b(n)$ and $q(n)$ are as in \cref{lem:qcma_inplace_lowerbound}.
\end{remark}

The only difference is that \cref{lem:parallel_hybrid} must work for classical oracles, which is exactly the hybrid method~\cite{BBBV97}.
Next, we show that any $\QMA^O$ query protocol can be parallelized by using Kitaev's proof that the Local Hamiltonian problem is \QMA-complete.~\cite{KSV02}.

\begin{restatable}{lemma}{lemQMAParallel}\label{lem:QMA-parallel}
    For any unitary $O$, $\QMA^{O} = \QMA^{\| O}$ and $\QMAo^{O} = \QMAo^{\| O}$. In terms of the number of rounds, $\QMA^{O}[R,1] \subseteq \QMA^{O}[1,O(R^4)]$ and $\QMAo^{O}[R,1] \subseteq \QMAo^{O}[1,O(R^4)]$.
\end{restatable}
\begin{proof}
    Let us first consider statement for \QMA. Consider some language $L \in \QMA^{O}[R,1]$ and for any $x\in \{0,1\}^n$, let $V$ be the verification circuit,
    \begin{align*}
        V = U_{R+1} O U_{R} \dots O U_1 \ket{\phi_x}\ket{0} = U_{L} U_{L-1} \dots U_2 U_1 \ket{\phi_x}\ket{0},
    \end{align*}
    where $L=2R+1$, for each even $i\in[L]$ $U_i$ is an oracle call and $\ket{\phi_x}$ is the witness.
    Recalling Kitaev's circuit-to-Hamiltonian construction~\cite{KSV02}, we may map $V$ to some Hamiltonian $H = \Hin + \Hprop + \Hout$ whose ground states are of the form,
    \begin{align*}
        \ket{\psi} = \frac{1}{\sqrt{L+1}} \sum_{t\in [0,L]} \left(U_t \dots U_0 \ket{\phi_x}\ket{0}\right) \otimes \ket{t}.
    \end{align*}
    Notice that the $\Hin$ term only checks whether the ancilla register is $\ket{0}$, but not the witness $\ket{\phi_x}$ and this instance of $H$ is no longer local. Nonetheless, this does not matter as with access to $O$, we may approximate its ground state energy. The main difference from~\cite{KSV02} is in the following term,
    \begin{align}
        \Hprop &= \sum_{t=1}^L H_t\\ 
        2H_t &= I\otimes\bigl(\ketbra{t}{t} + \ketbra{t-1}{t-1}\bigr) -  U_t\otimes\ketbra{t}{t-1} -  U_t^\dagger \ketbra{t-1}{t}.\label{eq:h_t_def}
    \end{align}
    We note that $H_t$ is a projector on the span of states $\ket{\psi}\ket{t-1} - (U_t\ket{\psi})\ket{t}$. Therefore, for a $d$-dimensional $U$,
    \begin{align*}
        H_t = \frac{1}{2}\sum_{j=1}^d \bigl(\ket{j}\ket{t-1} - (U_t\ket{j})\ket{t}\bigr)\bigl(\bra{j}\bra{t-1} - (\bra{j}U_t^\dagger)\bra{t}\bigr).
    \end{align*}
    To measure with regard to projector $H_t$, we may perform the following change of basis,
    \begin{align}
        U_t^\prime = U_t\otimes \ketbra{t-1}{t-1} + I\otimes \ketbra{t}{t},\label{eq:QMA-parallel:basis}
    \end{align}
    as then,
    \begin{align*}
        2 U_t^\prime H_t (U_t^\prime)^\dagger &= \sum_{j=1}^d \bigl((U_t\ket{j})\ket{t-1} - (U_t\ket{j})\ket{t}\bigr)\bigl((\bra{j}U_t^\dagger)\bra{t-1} - (\bra{j}U_t^\dagger)\bra{t}\bigr)\\
        &= \sum_{j=1}^d U_t\ketbra{j}{j}U_t^\dagger \otimes \bigl(\ket{t-1}-\ket{t}\bigr)\bigl(\bra{t-1}-\bra{t}\bigr)\\
        &= I \otimes \bigl(\ket{t-1}-\ket{t}\bigr)\bigl(\bra{t-1}-\bra{t}\bigr),
    \end{align*}
    which is an $X$-basis measurement on the $\{\ket{t-1}, \ket{t}\}$ subspace of the clock register. By \cref{lem:control-oracle,rem:permutation}, $U_t^\prime$ may be applied using a single application of $U_t$.

    The remainder of the analysis is standard. Letting $\tau$ be the test where we check a term of $H_t$ uniformly at random, we have that,
    \begin{align}
        \Pr[\tau\text{ accepts} | x\in L] \geq 1 - \frac{\tfrac{1}{3}}{(L+2)^2}\label{eq:prob_test} \\
        \Pr[\tau\text{ accepts} | x\notin L] \leq 1 - \frac{c}{(L+2)^4},\nonumber
    \end{align}
    where $c$ is a constant. Next, we run $O(L^4)$ parallel runs of the computation to estimate $\Pr[\tau \text{ accepts}]$, which allows us to decide whether $x\in L$. Note that this requires the prover to send multiple copies of $\ket{\psi}$ which may be entangled. However, this does not affect the protocol as it may be viewed as measuring one copy after another, and mixed states do not help by convexity.

    For \QMAo the same proof goes through, except the lower bound in \cref{eq:prob_test} is $1$. Therefore, running the test $O(L^4)$ and accepting if and only if all tests pass provides us with the desired soundness and completeness gap.
\end{proof}

Therefore, we may obtain a bounded adaptivity oracle separation using \cref{prob:perm_pointer_chasing} with respect to classical oracles. We emphasize that the restriction is important as allowing even a single additional round of queries puts the problem into \QCMA. Thus, even when $R(n)=n^k$, this is not an oracle separating \QMAo from unrestricted \QCMA.

\begin{restatable}[Lifting to unitary oracles]{theorem}{thmBoundedAdaptivity}\label{thm:bounded_adaptivity}
    For any fixed $R\in \poly(n)$ and $q(n) \leq 2^{\beta n}$ where $\beta \geq 0$, there exists a unary language $L$ and an oracle $O = \{O_n \}$ such that, $L\in \coRP^O[R,1]$, $L\in\QMAo^O[1,O(R^4)]$, but $L\notin \QCMA^O[R-1,q(n)]$.
\end{restatable}
\begin{proof}
    The oracle $O = \{O_n \}$ is a set of instances of \cref{prob:perm_pointer_chasing} where each $O_n$ encodes $\pi_n$ as in \cref{eq:classical_oracle_permutation}. The language $L$ is the unary language such that $1^n\in L$ if and only if $\pi_n$ is a YES instance. By \cref{rem:qcma_classical_lowerbound}, we have that there exists an $O$ such that $L\notin\QCMA^O[R-1, q(n)]$.

    Next, we argue the upper bounds by describing their corresponding algorithms given input $1^n$. For $L\in\coRP^O[R,1]$, consider the algorithm which samples a $j\in[N]$, computes $\pi^R(j)=f_{R+1}(j)=f_b(j)$ using $R$ queries (Notice that in \cref{prob:perm_pointer_chasing}, $R=b-1$) and accepts if and only if the output is even. In the yes-case, it accepts with certainty, while in a no-case, it accepts with probability at most $1/2$.
    For $L\in\QMAo^O[1,O(R^4)]$, by \cref{lem:QMA-parallel}, the same $R$-query algorithm run over a superposition over $[N]$ may be compressed to a single round of $O(R^4)$ queries.
\end{proof}

\subsection{Extensions of prior results}\label{sec:extensions}

Below, we focus on extending the results of~\cite{BHV26}. Let us first fix the key terminology. We will use $\lambda$ as the input size. Fix $q$ to be the smallest prime with $\lambda^5 < q \leq 2\lambda^5$ and set $C_\lambda \coloneq \Mult_{s,\FF_q,k}$ (an error-correcting code) with $k = \lambda^3$, $s = \lambda$, $\Sigma \coloneq \FF_q^s$.  For a function $H : [q] \times \Sigma \to \{0,1\}$. The relation $R_{C_\lambda, H}$ associated with the \emph{code intersection problem} is,
\begin{align*}
    R_{C_\lambda, H} \coloneq \left\{(x, v) \in \{0,1\}^q \times \Sigma^q : v \in C_\lambda,\ \forall j \in [q],\ H(j,v_j) = x_j\right\},
\end{align*}

Based on~\cite{YZ24}, it was shown that there exists a state $\ket{\adv_H}$ and an algorithm $\BiasYZ$ which, given some input $x$, finds a $v$ such that $(x,v)\in R_{C_\lambda, H}$ with high probability~\cite{BHV26}. In their proof, the function $H$ is drawn from a biased distribution, meaning that the random function drawn will output $1$ with probability $p^*$ as opposed to $\tfrac{1}{2}$.

We define the distribution $\Bias_{q,p,\Sigma}$ over functions $H$ as the product of $q$ distributions $\Bias_{p,\Sigma}$ over functions $H_i$ from $\Sigma$ to $\{0,1\}$ where $\Pr[H_i(y)=1]=p$. Let $\Good\subseteq \supp(\Bias_{q,p,\Sigma})$ be a subset which we define later.

\begin{restatable}[Adapted from Corollary~5.24 in~\cite{BHV26}]{corollary}{corPreciseYZ}\label{cor:precise_yz}
    Set $p^* \coloneq (\lambda-3)/(8\lambda^4)$.  For all sufficiently large $\lambda$, there exist an efficient quantum algorithm $\BiasYZ$ and $\poly(\lambda)$-qubit states $\{\ket{\adv_H}\}_H$ such that
    \begin{align}
        \Pr_{H\gets\Bias_{q,p^*,\FF_q^s}}\!\!&\left[\forall x\in\{0,1\}^\lambda,\;\Pr\!\left[(x\|0^{q-\lambda},v)\in R_{C_\lambda,H}: v\gets\BiasYZ(\ket{\adv_H},x)\right]\geq 1-2^{-\lambda}\right]\nonumber\\ &\geq  1-2^{-\lambda},\label{eq:biasyz_success}
    \end{align}
    and in particular, $\Pr_{H\gets\Bias_{q,p^*,\FF_q^s}}[H\in\Good]\geq 1-2^{-\lambda}$.
\end{restatable}

For the sake of completeness, we provide the proof with an adjusted parameter $p^*$ in \cref{sec:missing_proofs}.

\subsubsection{Zero error advice separation}

We first strengthen the advice separation result of~\cite{BHV26}. As $\ZQP \subseteq \BQP_1$, this affirmatively answers a variant of \cref{question1} where the witness is replaced with advice.

\begin{restatable}[Extension of Theorem~7.5 in~\cite{BHV26}]{theorem}{thmZqp}\label{thm:zqp}
    There exists a classical oracle $O$ such that $\ZQPqp^O \not\subseteq \BQPp^O$.
\end{restatable}
\begin{proof}
    We use the same oracle as~\cite{BHV26} and show that their $\BQPqp$ algorithm can be made zero-error as the lower bound against $\BQPp^O$ is exactly the lower-bound part of~\cite[Theorem~7.5]{BHV26}. Recall that the oracle is $O = \{O_\lambda\}$ where,
    \begin{align*}
        O_\lambda(x, v) \coloneq \begin{cases}
            G_\lambda(x) & \text{if } (x\|0^{q-\lambda}, v) \in R_{C_\lambda, H_\lambda},
            \\ \bot & \text{otherwise}
        \end{cases},
    \end{align*}
    where $\lambda\in\NN$, $x\in\{0,1\}^\lambda$, and $G_{\lambda}:\{0,1\}^\lambda\to\{0,1\}$ is a random function. The corresponding language used in the separation is $L^O \coloneq \bigsqcup_{\lambda \in \NN} G_\lambda^{-1}(1)$.

    The \ZQPqp algorithm runs the $\BiasYZ$ algorithm on the advice state $\ket{\adv_H}$ and input $x\in \{0,1\}^\lambda$ obtaining some $v$. Afterwards, it queries the oracle $O_\lambda(x,v)$, receiving some output $b$. If $b\in \{0,1\}$, then the verifier returns $b$. Otherwise, if $b=\bot$, they return \say{maybe}. Since $O_\lambda$ outputs $\bot$ if and only if $(x\|0^{q-\lambda},v)\not\in R_{C_\lambda,H_\lambda}$, the verifier never makes an error.  Note that by \cref{eq:biasyz_success} and the Borel-Cantelli lemma~\cite{Bor09,Can17}, good advice such that $\calA$ solves the problem for all $x\in\{0,1\}^{\lambda}$ exists for all but finitely many $\lambda$. For the instances where $\calA$ would fail, we can hardcode these inputs into a modified verifier $\calA^\prime$ which answers these instances correctly.
\end{proof}

\subsubsection{Precise verification}

The oracle separation of~\cite{BHV26} gives a separation between classes which have a constant completeness and soundness gap. We show that their argument extends to the setting where the gap is $2^{-\poly(n)}$ when this gap is fixed. The proof is the same, except we make a slight modification to the bias parameter $p^*$. The result uses the following problem.

\begin{problem}[Code Intersection Subset Size $\CISS_{\sigma}$]\label{prob:ciss}
    Let $H\in\Good$ be a random function, $\sigma:\NN \to \NN$ a function and $E \subseteq \{0,1\}^\lambda$. Given access to the following oracle,
    \begin{align*}
        O[H,E](x,v) = \begin{cases}
            1 &\text{if $x\in E$ and $(x,v)\in R_{C_\lambda, H}$.}\\
            0 &\text{otherwise,}
        \end{cases}
    \end{align*}
    decide whether (YES) $\abs{E} = 2^\lambda$ or (NO) $\abs{E} \leq 2^\lambda - \sigma(\lambda)$.
\end{problem}

The original problem~\cite{BHV26} set $\sigma(n)= \tfrac{2^{\lambda+1}}{3}$. Instead, we will use $\sigma(n)=2$. For simplicity of notation, in this section we will write that $\QCMA(2^{-n}) = \QCMA(1-2^{-n},1-2^{-n+1})$ and $\QMA(2^{-n}) = \QMA(1-2^{-n},1-2^{-n+1})$. Let us show that $\CISS_2$ is in $\QMA(2^{-n})$.

\begin{restatable}{lemma}{lemPreciseQmaAlgo}\label{lem:precise_qma_algo}
    $\CISS_2$ may be solved by $\QMA(2^{-n})$.
\end{restatable}
\begin{proof}
    The $\QMA(2^{-n})$ verifier $V$ randomly samples $x\leftarrow \{0,1\}^\lambda$, obtains $v$ by running the algorithm $v\leftarrow\BiasYZ(\ket{\adv_H}, x)$, queries $O[H,E](x,v)$ and outputs the result. In a YES instance, assuming $H\in \Good$, the probability of success depends solely on $\BiasYZ$ which is at least $1-2^{-\lambda}$. On the other hand, given a NO instance, even if the algorithm succeeds with probability $1$, $V$ may only accept if $x\in E$, which occurs with probability at most $\tfrac{2^\lambda - 2}{2^\lambda} = 1 - 2\cdot 2^{-\lambda}$.
\end{proof}

Next, we will show the $\QCMA(2^{-n})$ lower bound. The main argument in~\cite{BHV26} is that the existence of a \QCMA algorithm which makes $r$ queries and runs unitaries $\{U_i\}_{i\in [r]}$ implies the existence of a $\Guesser$ algorithm, a zero-query algorithm which can guess elements $(x,v)\in R_{C_\lambda,H}$ without making any queries to $O[H,E]$. In the algorithm below, we will use $O_S$ to denote the unitary which, given a set $S$, behaves as follows,
\begin{align*}
        O_S\ket{x,y} \coloneq  \ket{x, y\oplus \mathds{1}[x\in S]}.
    \end{align*}

\begin{algobox}{\Guesser}\alglabel{alg:guesser}
    \Input{String $1^l$ and unitaries $\{U_i\}_{i\in [r]}$.}
    \begin{enumerate}
        \item Sample a random witness $y\leftarrow \{0,1\}^{w}$ and let $S_0 = \emptyset$.
        \item Repeat the following for $i\in [l]$,
        \begin{enumerate}
            \item Randomly sample $j\leftarrow [0,r-1]$.
            \item Compute $\ket{\psi_j} = U_j O_{S_{i-1}}\dots O_{S_{i-1}} U_0 \ket{y,0}$.
            \item Measure the query register of $\ket{\psi_j}$ to get some $(x,v)$.
            \item Set $S_j \coloneq S_{j-1} \cup \{(x,v)\}$.
        \end{enumerate}
        \item Output $\{(x_{[1,\lambda]},v):(x,v)\in S_l\}$
    \end{enumerate}
\end{algobox}

The following statements are directly from~\cite{BHV26}, except we parametrized the soundness and completeness gap as $\delta$. The first is a lower bound on the success probability of \Guesser.

\begin{lemma}[\Guesser lower bound; Lemma~6.4 in~\cite{BHV26}]\label{lem:precise_guesser_lower}
    Let $\delta>0$ and suppose there exists a $\QCMA(c,c-\delta)$ algorithm $\calA$ that uses a $t(\lambda)$-bit classical witness, $Q(\lambda)$ oracle queries solving $\CISS_2$. Then for all $l\leq 2^\lambda-1$, there exists a no-query \Guesser such that for every $H\in\Good$,
    \begin{align*}
        \Pr\!\left[\begin{aligned}&\forall i\neq j,\;(x_i,v_i)\neq(x_j,v_j)\\&\wedge\;(x_i\|0^{q-\lambda},v_i)\in R_{C_\lambda,H}\end{aligned}:\{(x_i,v_i)\}_{i=1}^l\gets\Guesser(1^l)\right]\geq 2^{-t(\lambda)}\cdot\left(\frac{\delta^2}{16\,Q(\lambda)^2}\right)^l.
    \end{align*}
\end{lemma}

Below is the upper bound on the success of \Guesser, which will be used for a contradiction of the existence of $\calA$.

\begin{restatable}[\Guesser upper bound; Lemma~6.6 in~\cite{BHV26}]{lemma}{lemPreciseGuesserUpper}\label{lem:precise_guesser_upper}
    For all sufficiently large $\lambda$ and any algorithm \Guesser making no oracle queries, for all $l \leq 2^\lambda$,
    \begin{equation*}
        \Pr_{H\gets\Bias_{q,p^*,\FF_q^s}}\!\!\left[\begin{aligned}&\{(x_i,v_i)\}_{i=1}^l\gets\mathsf{Guesser}(1^l):\\&\forall i\neq j,\;(x_i,v_i)\neq(x_j,v_j)\text{ and }(x_i\|0^{q-\lambda},v_i)\in R_{C_\lambda,H}\end{aligned}\right] \leq 2^{-\lambda^2l/32}.
    \end{equation*}
\end{restatable}

Proof of \cref{lem:precise_guesser_upper} may be found in \cref{sec:missing_proofs}. With the results above, we may show the main result.

\begin{restatable}{theorem}{thmPreciseSeparation}\label{thm:precise_separation}
    There exists a classical oracle $O$ such that,
    \begin{align*}
        \QMA^{O}(2^{-n})\not\subseteq \QCMA^{O}(2^{-n}).
    \end{align*}
\end{restatable}
\begin{proof}
    Let us show hardness of solving $\CISS_2$ for any $ \QCMA^{O}(2^{-n})$ algorithm $\calA$. Suppose such an algorithm exists and it must succeed at least with probability $P_{\textup{low}}$. By \cref{cor:precise_yz,lem:precise_guesser_lower} with $\delta = 2^{-\lambda}$ and $l\coloneq t \leq a\lambda^a$ for some $a$,
    \begin{align*}
        P_{\textup{low}} \geq (1-2^{-\lambda}) 2^{-t} \cdot \left( \frac{2^{-2\lambda}}{16Q^2} \right)^t.
    \end{align*}
    By taking the logarithm, we have $\log_2 P_{\textup{low}} \geq -O(a^2 \lambda^{a+1})$. On the other hand, \cref{lem:precise_guesser_upper} gives us that the upper bound on the probability $P_{\textup{up}}$ is as follows,
    \begin{align*}
        \log_2 P_{\textup{up}} \leq \frac{-\lambda^2 t}{32} = \frac{-a\lambda^{a+2}}{32}.
    \end{align*}
    As $\tfrac{a\lambda^{a+2}}{32} \geq O(a^2 \lambda^{a+1})$ for large $\lambda$, we get the contradiction that $P_{\textup{low}} > P_{\textup{up}}$.

    The oracle $O = \{O_\lambda\}$ and the language $L$ is the unary language such that $1^\lambda\in L$ if and only if $O_\lambda$ is a YES instance. By \cref{lem:precise_qma_algo}, $L\in \QMA^{O}(2^{-n})$. On the other hand, due to the contradiction from the $\Guesser$ algorithm, there cannot exist a uniform $\QCMA^O(2^{-n})$ verifier which correctly evaluates all YES and NO instances. Therefore, we may diagonalize against all $\QCMA(2^{-n})$ machines, following~\cite[Appendix B]{BHV26}.
\end{proof}

Furthermore, using \cref{thm:precise_separation}, we find the following.

\begin{restatable}{corollary}{corPreciseSeparation}\label{cor:precise_separation}
    There exists a classical oracle $O$ such that,
    \begin{align*}
        \Pitwo^O \cap \PreciseQCMA^O \not\subseteq \QCMA^O(2^{-n}).
    \end{align*}
\end{restatable}
\begin{proof}
    First, notice that $\CISS_2$ is in $\Pitwo$ as in the YES case, for all $x$, there exists some $v$ such that $O[H,E](x,v) = 1$, while in the NO case, there exists an $x$ such that for all $v$, $O[H,E](x,v)=0$.

    Furthermore, it has been shown that $\PreciseQCMA = \mathsf{NP}^{\mathsf{PP}}$~\cite{MN17, GSSSY22}. As Toda's Theorem~\cite{Tod91} relativizes, $\Pitwo^O \subseteq \PreciseQCMA^O$ with respect to all oracles. Therefore, $\CISS_2$ is in $\PreciseQCMA$.
\end{proof}

Note that this holds due to the fact that unlike \QCMA, we do not know how to efficiently amplify the gap in \PreciseQCMA.

\subsubsection{Circuit complexity of sparse Hamiltonian ground states}

Suppose that, for some sparse Hamiltonian, we are given oracle access via $O_H^\rc$. Can we lower-bound the circuit complexity of preparing an approximate ground state of $H$? As we will show below, this is a straightforward consequences of the classical oracle separation between \QMA and \QCMA~\cite{BHNZ26,BHV26}. Nonetheless, to the best of our knowledge, this connection has not been established yet in prior literature, and shows an interesting parallel to the NLTS theorem~\cite{ABN23}.

In our proofs, we will assume that verifiers are using the gate set $\calG = \{\Xg, \CNOT, \Toff, \Had\otimes \Had\}$, which is computationally universal as it can simulate $\Had$~\cite{Shi03}.

\begin{restatable}{theorem}{thmSparseHamHard}\label{thm:sparse_ham_hard}
    There exists a family of $O(1)$-sparse Hamiltonians $\calH = \{ H_i \}_{i\in \NN}$ on $n_i>n_{i-1}$ qubits with integer entries bounded by $O(1)$ such that for all $i$, any state $\ket{\psi}$ generated by a circuit with $2^{o(n_i^{1/3} \log(n_i)^{-1})}$ gates (using $\calG$ and $O_{H_i}^\rc$),
    \begin{align*}
        \bra{\psi} H_{i} \ket{\psi} \geq \lambda_0(H_i) + \tfrac{1}{\poly(n_i)}.
    \end{align*}
\end{restatable}
\begin{proof}
    Let $L$ be the unary language and $O$ the oracle defined in~\cite[Theorem 6.8]{BHV26} such that $L\in \QMA^O$, but $L\notin \QCMA^O$. We construct $\calH$ by applying the circuit-to-Hamiltonian construction~\cite{KSV02} to the $\QMA^O$ verifiers $V_\lambda$, which use gates in $\calG$, for $\CISS_{2^{\lambda+1}/3}$ (The same algorithm as in \cref{lem:precise_qma_algo}). The authors showed that $V_\lambda$ uses a $\poly(\lambda)$-sized quantum witness, $\poly(\lambda)$ gates, workspace of $\Tilde{O}(\lambda^6)$ qubits, has completeness $c\geq 1-2^{-\lambda}$ and soundness $s\leq \tfrac{1}{3}$.
    
    Let $H\in \calH$ be a fixed instance $H=\Hin + \Hprop + \Hout$ created using the input $1^\lambda$ and oracle $O_\lambda$. We show that $H$ satisfies the conditions of the statement, starting with the fact it is $O(1)$-sparse. The only difference is that it the construction uses a binary clock instead of the standard unary clock in~\cite{KSV02}.\footnote{The issue with a unary clock arises for values that are larger than the number of steps $T$.} First, the terms $\Hin$ and $\Hout$ add $1$ non-zero entry. For $\Hprop$, assume that each unitary in $V_\lambda$ is $c$-sparse. From \cref{eq:h_t_def}, we can see that $H_t$ contributes $2(c+1)$ entries in rows with the clock basis $t$ and $t+1$. Notice that $\Had\otimes \Had$ is $4$-sparse, $\Xg, \CNOT, \Toff$ are $1$-sparse and the classical oracle $O_\lambda$, as it is a permutation matrix, is $1$-sparse. Therefore, $H$ is $O(1)$-sparse. Furthermore, the smallest entry in $\calG \cup \{O_\lambda\}$ is $\tfrac{1}{2}$, meaning that multiplying $H$ by $4$ ensures it is integer-valued.\footnote{Due to the $\tfrac{1}{2}$ in $H_t$, see \cref{eq:h_t_def}.} By the standard analysis on the energy of the Hamiltonian~\cite{KSV02}, we have,
    \begin{align*}
        1^\lambda \in L &\implies \lambda_0(H_{\lambda}) \leq 2^{-\Omega(\lambda)}\\
        1^\lambda \notin L &\implies \lambda_0(H_{\lambda}) \geq \frac{c}{\poly(\lambda)} \eqcolon c^\prime,
    \end{align*}
    where $c$ is a constant. 
    By using a single sparse projector for $\Hin$, we have that all entries of $H_\lambda$ are integers bounded by $O(1)$ in absolute value.
    Furthermore, given oracle access to $O_\lambda$, we may simulate queries to $O_H^{\rc}$ where each query to $O_H^{\rc}$ would require at most $1$ query to the original oracle $O[H,E]$ and $\poly(\lambda)$ gates.

    To prove the statement, we argue that a small circuit for an approximate ground state $\ket{\psi}$ contradicts the \QCMA lower bound of~\cite{BHV26}. Suppose that for all but finitely many YES instances $1^\lambda$, there exists a circuit $C_\lambda$ of size $s(\lambda)$ using $\calG \cup \{O_H^{\rc}\}$ gates that prepares a state $\ket{\psi_\lambda}$ such that,
    \begin{align*}
        \bra{\psi_\lambda} H_\lambda \ket{\psi_\lambda} \leq \tfrac{c'}3.
    \end{align*}
    Consider the following algorithm.

    \begin{algobox}{Energy estimation using circuit description}\alglabel{alg:energy_est}
        \Input{Oracle access to $O_\lambda$ and a string
    $y \in \{0,1\}^{s(\lambda)\cdot \log s(\lambda)}$ encoding $C_\lambda$.}
        \begin{enumerate}
            \item Letting $K=\poly(\lambda)$, run $C_\lambda\ket{0}$ $K$ times in parallel to prepare $K$ copies of $\ket{\psi}$, which requires at most $\poly(\lambda)\cdot s(\lambda)$ queries to $O_\lambda$ and gates.
            \item Run the algorithm from \cref{thm:qpe_energy} with $\epsilon = \tfrac{c'}{6}$ to obtain estimate $E$ of $\bra{\psi} H_n \ket{\psi}$, which requires $\poly(n)$ queries to $O_\lambda$.
            \item Accept if and only if $E <c'$.
        \end{enumerate}
    \end{algobox}
    Now, we show that the existence of \cref{alg:energy_est} contradicts \cref{lem:precise_guesser_lower,lem:precise_guesser_upper}. Assume that both the algorithm description and length are bounded by $s \leq 2^{\Omega(\lambda^2)} - 1$. As the completeness and soundness gap is $\tfrac{1}{3}$, \cref{lem:precise_guesser_lower} implies that,
    \begin{align*}
        \log_2 P_{\textup{low}} \geq -s - s \left(8 + 2\log Q\right) = -s\left(9 + 2\log Q\right).
    \end{align*}
    On the other hand, \cref{lem:precise_guesser_upper} implies that,
    \begin{align*}
        \log_2 P_{\textup{up}} \leq -\frac{\lambda^2 s}{32}.
    \end{align*}
    Therefore, to obtain the contradiction that $P_{\textup{low}} > P_{\textup{up}}$, we need $9 + 2\log_2 Q < \tfrac{\lambda^2}{32}$, which is satisfied when $Q < 2^{\Omega(\lambda^2)}$.
    Lastly, the size of workspace of the $\QMA^O$ verifier is $O(\lambda^6 \log \lambda)$, most of which is the size of the quantum witness. Therefore, $n_i = \tilde{O}(\lambda^6 \log \lambda)$, meaning that every circuit $C$ preparing $\ket{\psi}$ must exceed size $s(n) \leq 2^{o(n_i^{1/3} \log(n_i)^{-1})}$.
\end{proof}

We note that one may obtain a similar result using the classical oracle separation in~\cite{BHNZ26}. 
With our work, we can extend \cref{thm:sparse_ham_hard} to frustration-free Hamiltonians in the bounded adaptivity parallel-query setting. Note that extending \cref{thm:sparse_ham_hard} to the frustration-free setting without query restrictions would imply a classical oracle separation between \QCMA and \QMAo, an open problem.

\begin{restatable}{theorem}{thmBoundedAdaptivityFrustration}\label{thm:bounded_adaptivity_frustration}
    For any $k\in\NN$, there exists a family of $O(1)$-sparse frustration-free Hamiltonians $\calH = \{ H_i\}_{i\in \NN}$ on $n_i>n_{i-1}$ qubits with integer entries bounded by $O(1)$, such that for all $i$, any state $\ket{\psi}$ generated by a circuit with $2^{o(n_i)}$ gates (using $\calG$ and $O_H^\rc$) using $n^k$ rounds of queries to $O_H^\rc$ has
    \begin{align*}
        \bra{\psi} H_i \ket{\psi} \geq \tfrac{1}{\poly(n_i)}.
    \end{align*}
    Additionally, the ground state of $H_i$ may be prepared using $n_i^k + O(1)$ adaptive queries to $O_H^{\rc}$.
\end{restatable}
\begin{proof}
    Let $L$ be the language described in \cref{thm:bounded_adaptivity} with $R = n_i^k + O(1)$ and $O$ the corresponding oracle.
    Here the $O(1)$ term is the constant query complexity of the kernel testing algorithm in \cref{lem:kernel-tester}.
    Let $V_n$ be the verifiers for all the instances $1^n\in L$. Similarly to \cref{thm:sparse_ham_hard}, the family $\calH$ is constructed by applying the circuit-to-Hamiltonian construction with a binary clock on these $V_n$. Therefore, each $H\in\calH$ is $O(1)$-sparse and has constant integer entries. Furthermore, as they are constructed only for YES instances, there exists a state making $V_n$ accept with perfect completeness and thus $H$ is frustration-free. 
    As before, the Hamiltonian is $O(1)$-sparse with integer entries bounded by $O(1)$.
    
    For the sake of contradiction, assume that for all but finitely many $H\in \calH$, there exists a circuit $C_n$ of size $s(n)$ which uses gates in $\calG \cup O_H^\rc$, has $n^k$ adaptive rounds of queries to $O_H^\rc$ and prepares the state $\ket{\psi}$ whose energy is less than $c'/n$. Consider the following algorithm.

    \begin{algobox}{Energy estimation using circuit description}\alglabel{alg:frustration_free}
        \Input{Oracle access to $O_n$ and a string $y\in \{0,1\}^{s(n)\cdot \log s(n)}$ encoding $C_n$.}
        \begin{enumerate}
            \item Let $K=\poly(n)$. Run $C_n\ket{0}$ $K$ times in parallel to prepare $K$ copies of $\ket{\psi}$, which requires at most $n^k$ rounds of $\poly(n)\cdot s(n)$ queries to $O_n$, and $\poly(n)\cdot s(n)$ gates in total.
            \item Run the algorithm from \cref{lem:kernel-tester} to obtain a constant-factor estimate $E$ of $\bra{\psi}H_n\ket{\psi}$, which uses $O(1)$ rounds of $\poly(n)$ queries to $O_n$.
            \item Accept if and only if $E<2 c'\alpha/n$.
        \end{enumerate}
    \end{algobox}

    Note that \cref{alg:frustration_free} will accept with probability $1-1/\poly(n)$ in the YES case, assuming the state $\ket{\psi}$ has energy at most $c'/n$. This follows from Hoeffding's inequality in the bounds on the rejecting probability in \cref{lem:kernel-tester}.
    Additionally, in the NO-case, the energy is always at least $c'$, and so \cref{alg:frustration_free} will reject with probability at least $1-1/\poly(n)$ in the NO case.

    Next, we show that for $s(n) \leq 2^{o(n_i)}$, \cref{alg:frustration_free} ($O_{n_i} = O_{H_i}^\rc$) leads to a contradiction. Per \cref{lem:qcma_inplace_lowerbound}, we have that $L\not\in \QCMA^O$ with $n^k$ adaptive rounds of queries to $O$. Furthermore, the restriction to witnesses-size $w=\poly(n)$ may be relaxed. This appears in \cref{eq:final_bound}, where we need that,
    \begin{align*}
        2^{2 + 2\alpha n + w} < 2^{\frac{2(1-2\alpha)}{256}}.
    \end{align*}
    Taking the logarithm, this condition becomes,
    \begin{align*}
        w < c\cdot 2^{(1-2\alpha)n}.
    \end{align*}
    Therefore, \cref{lem:qcma_inplace_lowerbound} holds even when the witness has size $w = o(2^{(1-2\alpha)n})$.
    Hence, the existence of an efficient circuit $C_n$ with $\tfrac{w}{\log w}$ gates implies $\QCMA^O$ verifier which runs \cref{alg:frustration_free} which uses $n^k$ rounds of queries to $O$, contradicting \cref{lem:qcma_inplace_lowerbound} with $R=n^k+O(1)$.
    Therefore, there does not exist a circuit of size $2^{o(n)}$. On the other hand, the zero-energy state $\ket{\psi}$ can be prepared by preparing $\ket{\psi_{n^k}}$, where $\ket{\psi_0}\coloneq \ket{[N]}$ and $\ket{\psi_i} \coloneq O_n\ket{\psi_{i-1}}$, followed by turning $\ket{\psi_{n^k}}$ into a history state.
\end{proof}

\newpage

\bibliographystyle{alpha}
\bibliography{main}

\appendix

\section{Postponed Proofs}\label{sec:missing_proofs}

\corPreciseYZ*
\begin{proof}
    As the proof is a slight adjustment of the original statement, we use notation directly from~\cite[Section~5]{BHV26}.
    We apply~\cite[Theorem~5.22]{BHV26} with the same parameters, except the bias is $p^* = (\lambda-3)/(8\lambda^4)$. Let us show that $(C_\lambda, \Dec_{\lambda^\perp})$ is $(p^*, \mu, S)$-good where $\mu(\lambda) \leq e^{-\lambda^2/96}$ and $S= \{0,1\}^{\lambda}\times 0^{q-\lambda}$. Fix some $x\in \{0,1\}^\lambda$ and draw $e\sim \mathcal{D}_{p^*, x\|0^{q-\lambda}}$. Let $X\coloneq \hw(e_{\lambda+1:q})$. Notice that for sufficiently large $\lambda$, $\E[X] = (q-\lambda)p^* \leq 2\lambda^5 \tfrac{\lambda-3}{8\lambda^4}$ and $\E[X] > \tfrac{\lambda^2}{24}$. Letting $T = 2\E[X]$, by the Chernoff bound, we have, $\Pr[X > T] \leq e^{-\E[X]/4} \leq e^{-\lambda^2/96}$. When $X<T$, we have that $\hw(e) \leq \lambda + T < \tfrac{\lambda^2}{2} \leq \tfrac{\textup{dist}(C_\lambda^\bot)}{2}$. Hence the unique decoding in \BiasYZ succeeds when $X\leq T$. Furthermore, notice that $2q(1-p^*)^{\abs{\Sigma}} \leq 4\lambda^5 \cdot e^{-p^* \lambda^{5\lambda}} \leq 2^{-2\lambda - 2}$.

    Therefore, applying~\cite[Theorem~5.22]{BHV26}, we find that for each fixed $x\in \{0,1\}^\lambda$,
    \begin{align*}
        \Pr_{H \gets \Bias_{q,p^*,\FF_q^s}}\!\!\left[\Pr\!\left[(x\|0^{q-\lambda},v)\in R_{C_\lambda,H}: v \gets \BiasYZ(\ket{\adv_H},x)\right] \ge 1-2^{-\lambda} \right]
        \ge 1 - 2^{-2\lambda - 1}.
    \end{align*}
    Therefore, the result follows by union-bounding over all $x\in \{0,1\}^\lambda$.
\end{proof}

\lemPreciseGuesserUpper*
\begin{proof}
    Fix $l$ distinct codewords $v_1,\dots v_l\in C_\lambda$ and for $j\in[q]$, let $S_j = \{\sigma\in \Sigma: \exists i\in [l], (v_i)_j=\sigma\}$. By $(l,2l)$-list-recoverability property of $C_\lambda$, $\sum_j \abs{S_j} \geq ql/2 \geq \lambda^5l/2$. As \Guesser does not make any queries to the oracle, it must correctly guess at least $\lambda^5l/2$ values $H(j,\sigma)$, each with probability $1-p^*$. As $p^*\lambda^5 \geq \lambda^2/16$, we have that the probability of success is at most $(1-p^*)^{\lambda^5 l/2} \leq e^{-p^* \lambda^5 l/2} \leq 2^{-\lambda^2 l/32}$.
\end{proof}

\end{document}